\documentclass[letterpaper, 11pt, onecolumn]{article}
\pdfoutput=1 %

\usepackage{lmodern,booktabs,authblk,tocloft}
\usepackage[margin=1in]{geometry}
\usepackage{graphicx}
\usepackage[linesnumbered,ruled,vlined]{algorithm2e}
\SetKwInput{KwInput}{Input}                %
\SetKwInput{KwOutput}{Output}
\SetKwInOut{KwPromise}{Promise}
\SetKwFor{RepTimes}{repeat}{times}{end}
\usepackage{complexity}
\usepackage{amsmath}
\usepackage{fullpage}
\usepackage{amsfonts}
\usepackage{amssymb}
\usepackage{dsfont}
\usepackage{amsthm}
\usepackage{tensor}
\usepackage[dvipsnames]{xcolor}
\usepackage{comment}
\usepackage{mathtools}
\usepackage{physics}
\usepackage{bbold}
\usepackage{appendix}
\usepackage{multirow}
\usepackage[ruled,vlined]{algorithm2e}
\usepackage[colorlinks]{hyperref}
\hypersetup{
pdfstartview={FitH},
pdfnewwindow=true,
colorlinks=true,
linkcolor=MidnightBlue,
citecolor=MidnightBlue,
filecolor=MidnightBlue,
urlcolor=MidnightBlue}
\usepackage{enumitem}
\usepackage[capitalize]{cleveref}

\makeatletter
\DeclareRobustCommand\bfseries{%
  \not@math@alphabet\bfseries\mathbf
  \fontseries\bfdefault\selectfont\boldmath}
\makeatother

\def\bm{{\bf m}}

\def\ri{{\rm i}}

\def\cB{\mathcal{B}}

\def\cH{\mathcal{H}}

\def\cO{\mathcal{O}}

\def\one{{\mathchoice {\rm 1\mskip-4mu l} {\rm 1\mskip-4mu l} {\rm
1\mskip-4.5mu l} {\rm 1\mskip-5mu l}}}
\def\dqc1{\textsc{DQC1}}
 \newcommand{\sket}[1]{| #1 \rangle}        %

\newcommand{\tshostate}[1] {\ket{\psi_{#1}^{\textrm{c}}}}%

\newcommand{\dshostate}[1] {\ket{\psi_{#1}}} %
\newcommand{\bradshostate}[1]{\bra{\psi_{#1}}} %

\newcommand{\position}{\hat{x}}
\newcommand{\momentum}{\hat{p}}

\newcommand{\discreteposition}{\overline x}
\newcommand{\discretemomentum}{\overline p}
\newcommand{\discretehamiltonian}{\overline{\mathrm{H}}}

\makeatletter
\newtheorem*{rep@theorem}{\rep@title}
\newcommand{\newreptheorem}[2]{%
\newenvironment{rep#1}[1]{%
 \def\rep@title{#2 \ref{##1}}%
 \begin{rep@theorem}}%
 {\end{rep@theorem}}}
\makeatother

\newtheorem{theorem}{Theorem}[section]
\newreptheorem{theorem}{Theorem}
\newtheorem{lemma}[theorem]{Lemma}
\newreptheorem{lemma}{Lemma}
\newtheorem{claim}[theorem]{Claim}
\newreptheorem{claim}{Claim}

\newreptheorem{definition}{Definition}

\newreptheorem{remark}{Remark}
\newtheorem{corollary}[theorem]{Corollary}
\newreptheorem{corollary}{Corollary}

\newreptheorem{observation}{Observation}

\def\rd{{\rm d}}

\def\ri{{\rm i}}

\title{Efficient quantum circuits for high-dimensional representations of $SU(n)$ and Ramanujan quantum expanders}

\author[1]{Vishnu~Iyer}
\author[1,2]{Siddhartha~Jain}
\author[2]{Stephen~Jordan}
\author[2]{Rolando~D.~Somma}

\affil[1]{\footnotesize The University of Texas at Austin, Austin, TX 78712, United States}
\affil[2]{\footnotesize Google Quantum AI, Venice, CA 90291, United States}

\begin{document}

\maketitle

\begin{abstract}
We present efficient quantum circuits that implement high-dimensional unitary irreducible representations (irreps) of $SU(n)$, where $n \ge 2$ is constant. For dimension $N$ and error $\epsilon$, the number of quantum gates in our circuits is polynomial in $\log(N)$ and $\log(1/\epsilon)$. Our construction relies on the Jordan-Schwinger representation, which allows us to realize irreps of $SU(n)$ in the Hilbert space of 
$n$ quantum harmonic oscillators. Together with a recent efficient quantum Hermite transform, which allows us to map the computational basis states to the eigenstates of the quantum harmonic oscillator, this allows us to implement these irreps efficiently. Our quantum circuits can be used to construct explicit Ramanujan quantum expanders, a longstanding open problem.
They can also be used to fast-forward the evolution of certain quantum systems.

\end{abstract}

\maketitle

\newpage
\tableofcontents
\newpage
\section{Introduction}

The special unitary groups $SU(n)$, $n \ge 2$, are a cornerstone of both physics and quantum information theory (Cf.~\cite{georgi2000lie,li2025d}). On the one hand these groups determine the fundamental symmetries of matter: $SU(2)$, generated by the algebra of angular momentum operators, governs the physics of spin systems and the electroweak interaction~\cite{hollik2014quantum}, and $SU(3)$ underpins Quantum Chromodynamics (QCD) and arises naturally in nuclear physics~\cite{arima1999elliott}. 
On the other hand, the groups $SU(n)$ are critical for investigating the complexity of Hilbert space: specific subsets of high-dimensional $SU(n)$ unitaries act as quantum expanders, which are operations that rapidly delocalize information across a system~\cite{lubotzky1988ramanujan,harrow2007quantum}. These ``fast scramblers''  are essential for modeling quantum thermalization, chaos, and black-hole dynamics~\cite{sekino2008fast}, and are intrinsically related with   unitary designs~\cite{harrow2009random}. Consequently, studying efficient quantum implementations of these high-dimensional unitaries is a priority for 
quantum information, having important fundamental and practical applications in
quantum simulation and beyond, as we will discuss.

Formally, the $SU(n)$ group is spanned by $n \times n$
unitary matrices of determinant 1. 
These matrices determine the {\em fundamental} representation
of $SU(n)$.
Nevertheless, the structure of the group becomes more relevant
when considering higher-dimensional representations,
as they allow for richer models of physics (e.g., a spin other than 1/2) and other essential quantum information tasks. To describe such representations, we start from the definition of $SU(n)$ 
via its  Lie algebra ${\mathfrak {su}}(n)$.
In the physicist's convention, 
the fundamental representation for  ${\mathfrak {su}}(n)$ is given by $n \times n$
Hermitian matrices $\{T_1,\ldots, T_{n^2-1}\}$ of trace zero that satisfy the commutation relations
\begin{align}
\label{eq:abstractcommutation}
    [T_a, T_b]=T_a T_b - T_b T_a =\ri \sum_{c=1}^{n^2-1} f_{abc} T_c \;.
\end{align}
The coefficients $f_{abc} \in \mathbb R$ for $a,b,c \in \{1,\ldots,n^2-1\}$ are totally antisymmetric and constitute the  structure factors of ${\mathfrak {su}}(n)$. Equation~\eqref{eq:abstractcommutation} is a main property of a Lie algebra, i.e., commutation relations that are closed.
Any element $U$ in the fundamental representation of $SU(n)$ can then be specified with $n^2-1$ real parameters $\theta_a$, so that
\begin{align}
\label{eq:UinSU(n)}
     U = \exp \left(\ri \sum_{a=1}^{n^2-1} \theta_a T_a \right) \in \mathbb C^{n \times n}\;.
\end{align}
It suffices to restrict $\theta_a \in [0, 4 \pi)$ in our convention.

In general, it is useful to  consider an {\em abstract} definition
of a Lie algebra $\mathfrak g$, and also of its Lie group $G$, that depends solely on the structure factors rather than a specific matrix representation.
In our case, we can specify $\mathfrak g \equiv \mathfrak {su}(n)$ via a basis $\{\tilde T_1, \ldots , \tilde T_{n^2-1} \}$, together with a Lie bracket $[.,.]: \mathfrak g \times \mathfrak g \mapsto \mathfrak g$: 
\begin{align}
    [\tilde T_a, \tilde T_b] = \ri \sum_{c=1}^{n^2-1} f_{abc} \tilde T_c \;.
\end{align}
The Lie bracket is bilinear, antisymmetric, and satisfies the Jacobi identity.
This is similar to Eq.~\eqref{eq:abstractcommutation},
only that we are not alluding to any particular matrix representation of the vectors or elements $\tilde T_a$.
That is, while the Lie bracket specifies the abstract Lie algebra, it is not a specific formula like the commutator of Eq.~\eqref{eq:abstractcommutation} since the $\tilde T_a$ are abstract elements. 
Then, in the context of representation theory, 
any abstract element $\tilde T_a$ is identified with an $N \times N$-dimensional and Hermitian matrix $T_a$, such that the {\em commutator}, defined by
   $ [T_a, T_b]:= T_a T_b - T_b T_a$ satisfies $[T_a, T_b] = \ri \sum_{c=1}^{n^2-1} f_{abc} T_c$.
   That is, the Lie bracket is replaced by the commutator
   between matrices in the representation, and both satisfy similar rules. Since now the dimension can satisfy $N \ge n$ in general, this allows for higher dimensional representations of 
   the abstract Lie algebra $\mathfrak{su}(n)$,
   which is the case we consider here. It also allows for higher-dimensional representations of the group $SU(n)$ by simply letting $U = \exp(\ri \sum_a \theta_a T_a) \in \mathbb C^{N \times N}$.
   See Ref.~\cite{fuchs1995affine} for 
   the basics of Lie algebras and Lie groups.

A simple example to illustrate representation theory is that of $SU(2)$, where the abstract Lie algebra $\mathfrak{su}(2)$ is spanned by elements $\tilde T_1$, $\tilde T_2$, and $\tilde T_3$, which are related to the angular momentum operators $J_x$, $J_y$, and $J_z$ of a quantum system. That is,
the fundamental representation of $\mathfrak{su}(2)$ is that of ``spin-1/2'' and obtained by half the Pauli matrices as $T_a=\frac 1 2 \sigma_a \in \mathbb C^{2 \times 2}$. The next ``spin-1'' representation is given by 
\begin{align}
     T_1 = \frac 1 {\sqrt 2}\begin{pmatrix}
        0 & 1 & 0 \cr 1 & 0 & 1 \cr 0 & 1 & 0
    \end{pmatrix},\;
     T_2 = \frac 1 {\sqrt 2} \begin{pmatrix}
        0 & -\ri & 0 \cr \ri & 0 & -\ri \cr 0 & \ri & 0
    \end{pmatrix}, \;
    T_3 = \begin{pmatrix}
        1 & 0 & 0 \cr 0 & 0 & 0 \cr 0 & 0 & -1
    \end{pmatrix} .
\end{align}
Of particular interest in this work is the case where the dimension of the representation is $N \gg n$. Indeed, we will mainly consider cases where $n$ is a constant (e.g., $n=2$ or $n=3$), and where $N$ can increase exponentially in some problem parameter $r$ (e.g., $N=2^r)$\footnote{Our results can be extended to the case where $n=\polylog(N)$ and can be directly applied to subgroups of $SU(n)$.}. We will also consider {\em faithful} representations of $\mathfrak{su}(n)$, meaning that there is a one-to-one correspondence between each basis element $\tilde T_a$ and its matrix representation $T_a$, and {\em irreducible} representations or irreps, where the matrix representations cannot be decomposed as direct sums of smaller matrices (i.e., there is no invariant subspace other than the full space $\mathbb C^N$).

Noting that high-dimensional representations 
of $SU(n)$ are unitary matrices, it is natural to consider their implementations and applications in quantum algorithms.
That is, for a unitary   $U = \exp \left(\ri \sum_{a=1}^{n^2-1} \theta_a T_a \right) \in \mathbb C^{N \times N}$ in a representation of the group $SU(n)$, we can think of $U$ as a quantum circuit acting on a Hilbert space $\mathbb C^N$ of, say, $r=\log_2 N$ qubits.
For quantum information tasks, it is desirable 
that such quantum circuit's complexity has a 
favorable scaling, particularly in terms of the dimension $N$ and precision $\epsilon$.  
This has been a longstanding open problem,
and constitutes our main result. That is, we will show that $U$ can be implemented within arbitrary error  $\epsilon >0$  with a quantum circuit of complexity polynomial in $\log(N)$ and $\log(1/\epsilon)$.  

Our result does not trivially follow from standard quantum algorithms. For example, while $U$ can be thought of as an evolution of a quantum system with Hamiltonian $\sum_a \theta_a T_a$, the Hermitian matrices in the representation satisfy $\|T_a\|=\Theta(N)$, and hence generic methods for Hamiltonian simulation~\cite{BACS07,berry2015simulating,low2017optimal} would result in complexity that is polynomial in $N$ rather than $\log(N)$. 
However, our main result shows
that these $SU(n)$ operations can be {\em fast-forwarded}. 
By fast-forwarding we mean any quantum simulation approach that bypasses the no-fast-forwarding theorems, which would give complexity $\Omega(t \|H\|)$ for simulating a Hamiltonian $H$ for time $t$~\cite{BACS07,atia2017fast,HHK21,gu2021fast}.

Our solution approach
considers oscillator realizations of $SU(n)$. Essentially, it is possible to construct an $N \times N$-dimensional irrep of $\mathfrak {su}(n)$ by considering specific operators acting on the space of $n$ quantum harmonic oscillators or bosonic modes. More importantly, we show how the corresponding $SU(n)$
unitaries in these bosonic systems can be implemented efficiently, in time polynomial in $\log (N)$ and $\log(1/\epsilon)$.
 This, however, does not directly solve our problem, since we have moved from the space $\mathbb C^N$ of $r$ qubits to that of the $n$ harmonic oscillators,
 and we would need an isometry that allows us to transform between spaces efficiently. Remarkably, 
we are able to construct this isometry by
leveraging a recent efficient quantum Hermite transform~\cite{jain2025efficient}, that maps basis states into Hermite states. These Hermite states are precisely the eigenstates of the quantum harmonic oscillator, also known as Fock states~\cite{griffiths2018introduction}.
Accordingly, our quantum circuit to implement $U$ acts on an enlarged space that is needed to accurately represent the quantum harmonic oscillators.

In essence, our approach shares some similarities with Ref.~\cite{zalka2004implementing}. There, a quantum algorithm for implementing a high-dimensional $SU(2)$ unitary is proposed.
Their idea is to use spherical harmonics rather than the Hermite states. Spherical harmonics also give rise to irreducible representations of $SU(2)$.
However, the runtime of the algorithm in Ref.~\cite{zalka2004implementing} is not carefully analyzed, and
that approach is not known to be efficient -- see Sec.~\ref{sec:relatedwork}.  We also remark that, in our case, the dimension $N$ cannot be arbitrary and needs to obey the relation $N=\binom{M+n-1}{n-1}$, where $M \ge 0$ is an arbitrary integer.
This is because the oscillator representation of $SU(n)$
only produces the ``totally symmetric'' irreps, constraining the possible dimensions\footnote{These symmetric irreps are characterized by Young tableaux of a single row.}. 
Nevertheless, this poses no problem for the applications we consider and $N$ can still be arbitrarily large. Additionally, there is no constraint for $N$ when $n=2$, in which case all irreps are totally symmetric.

An immediate consequence   is that
our efficient quantum algorithms for $SU(2)$
unitaries can be used to construct explicit Ramanujan quantum expanders.
This was also a longstanding problem in quantum information. It is implied by Refs.~\cite{lubotzky1988ramanujan,harrow02expander,harrow2007quantum},
where it is shown that there exists a finite and explicit set of unitaries in $SU(2)$ that give rise to an optimal quantum expander, where the spectral gap in terms of the degree is the largest possible.
Such quantum expanders are then optimal scramblers of quantum information.
We can then construct these expanders with arbitrary accuracy $\epsilon >0$ efficiently. 
Another direct consequence is the ability to fast-forward the evolution of certain quantum systems,
such as  the quantum kicked top in quantum chaos, adding to the literature
of physics models that can be fast-forwarded~\cite{gu2021fast}.
Ultimately, we expect our quantum primitives to also open the door to other novel applications.

Our paper is then organized as follows. In Sec.~\ref{sec:problemstatement}
we define and formalize the 
``high-dimensional $SU(n)$ problem'', which seeks an efficient implementation of $U \in SU(n)$,  and summarize the main results and corollaries in Sec.~\ref{sec:mainresults}.
In Sec.~\ref{sec:relatedwork}
we comment on related work, including Ref.~\cite{zalka2004implementing}, and also comment
on efficient quantum algorithms for the quantum Schur
transform that address a different but related problem.
In Sec.~\ref{sec:oscillatorrep} we describe how the   high-dimensional irreps of $SU(n)$ can be obtained via the oscillator or Jordan-Schwinger representations, that is, via $n$ quantum harmonic oscillators or bosonic systems. There, we also show how the corresponding $SU(n)$  unitaries can be fast-forwarded and efficiently implemented. Moreover, since the quantum harmonic oscillators are defined in the continuum and an infinite dimensional space, we need to approximate these via discretization and hence describe their finite-dimensional versions in Sec.~\ref{sec:finiteQHO}. In particular, in Sec.~\ref{sec:FFdiscreteQHO} we show that it is possible to consider a sufficiently high-dimensional discrete quantum harmonic oscillator to reproduce the results in the continuum with arbitrary accuracy. Next, we   construct the isometry that maps between spaces in Sec.~\ref{sec:quantumSU(n)}, and  also provide the full efficient quantum algorithm for the high-dimensional $SU(n)$ unitary there.
We discuss potential applications to quantum expanders and fast forwarding in Sec.~\ref{sec:applications} and conclude in Sec.~\ref{sec:conclusions}.

\subsection{Problem statement}
\label{sec:problemstatement}

We formalize the problem of implementing a high-dimensional irrep of $SU(n)$, which we refer to as the high-dimensional $SU(n)$ problem,
in a way that relates directly with a totally symmetric or oscillator representation for clarity.
That is, there are infinitely-many representations of dimension $N$ (e.g., it is always possible to conjugate the $T_a$'s with an arbitrary unitary and obtain another representation), but we will choose the standard ones that are determined by the Jordan-Schwinger mapping. This guarantees that the matrices $T_a$ are sparse and have a simple structure, among other features.

For $n \ge 2$ and $N \ge n$, we let $U \in \mathbb C^{N\times N}$ be a unitary in the corresponding $N$-dimensional irrep  of $SU(n)$. The goal is to approximate $U$ with a quantum circuit using as few gates as possible when $N \gg 1$. To this end, we need to be careful in the way that $U$ is specified.
Ultimately, we will describe $U$ as the exponential of matrices that form a representation of the Lie algebra $\mathfrak{su}(n)$ and, to that end, we will use the standard Cartan-Weyl basis. 
In this basis, the Cartan subalgebra $\mathfrak g_D$
is spanned by the vectors $\{\tilde H_i\}_{i}$, for $1 \le i \le n-1$, which correspond to a maximal set of 
elements where the Lie bracket vanishes, i.e., 
$[\tilde H_i , \tilde H_{i'}]=0$ for all $i$ and $i'$. 
We also define certain ladder elements or ``root'' vectors $\{ \tilde E_{j,k}\}_{ j ,k  }$, for $j,k \in \{1,\ldots,n\}$ and $j \ne k$, and  
the corresponding Hermitian versions $\tilde S_{j,k}:=\frac 1 2 (\tilde E_{j,k}+(\tilde E_{j,k})^\dagger)$ and $\tilde A_{j,k}:=\frac \ri 2(\tilde E_{j,k}-(\tilde E_{j,k})^\dagger)$, where $(\tilde E_{j,k})^\dagger =\tilde E_{k,j}$.
The Lie bracket for $\mathfrak {su}(n)$ is $[\tilde E_{j,k}, \tilde E_{l,m}]=\delta_{kl} \tilde E_{j,m}-\delta_{jm}\tilde E_{l,k}$ and $[\tilde E_{j,k}, \tilde H_i]$ is proportional to $\tilde E_{j,k}$.
The total number of operators is then $(n-1) + 2 n(n-1)/2=n^2-1$, which matches the dimension of $\mathfrak{su}(n)$; indeed, a basis of the abstract Lie algebra $\mathfrak{su}(n)$ is given by 
 $\{\tilde H_i, \tilde S_{j,k}, \tilde A_{j,k}\}$.
 
In the totally symmetric irrep of $\mathfrak{su}(n)$, which is the one we consider, the dimension obeys 
\begin{align}
    N =\binom{M + n -1}{ n-1} \;,
\end{align}
for arbitrary integer $M \ge 0$. That is, $N$ cannot be arbitrary for $n \ge 3$, but there is no constraint on $N$ for the case $n=2$.
The representations of the $\tilde H_i$'s are given by diagonal matrices, and the representations of the
$\tilde E_{j,k}$'s  are given by 1-sparse matrix
where the nonzero entries are off-diagonal. 
Specifically, 
we can construct the desired irrep
by considering a proper labeling of the basis vectors
of $\mathbb C^N$.
The standard way is to use the descending lexicographic ordering of size $n$: $\ket 0 \equiv \ket{M , 0 , 0 , \ldots , 0 },  \ket 1 \equiv \ket{M-1 , 1 , 0,\ldots , 0 } , \ldots ,  \ket {M+1} \equiv  \ket{ 0 , M , 0 , \ldots , 0}, \ket {M+2} \equiv  \ket{ 0 , M-1 , 1 , \ldots , 0}, \ldots , \ket{N-1}  \equiv  \ket{0,0,0,\ldots,M}$.
For $i \in \{1,\ldots,n\}$, let $m_i (\ell) \in \{0,\ldots,M\}$ be the value of the $i^{\rm th}$ entry of $\ket {\ell} \equiv \ket{m_1(\ell),\ldots,m_n(\ell)}$, where $0 \le \ell \le N-1$. 
We also have the inverse mapping $\bm^{-1}$ such that $\bm^{-1}(m_1(\ell),\ldots,m_n(\ell))=\ell \in \{0,\ldots,N-1\}$.

Then, the totally symmetric irrep 
of $\mathfrak{su}(n)$ is given by matrices that act on basis vectors as
\begin{align}
\nonumber
    E_{j,k} \ket{\ell } &\equiv  E_{j,k}\ket{m_1(\ell),\ldots,m_j(\ell),\ldots,m_k(\ell),\ldots,m_n(\ell)} \\
    \nonumber
    & =
    \sqrt{(m_j(\ell)+1)m_k(\ell)} \ket{m_1(\ell),\ldots,m_j(\ell)+1,\ldots,m_k(\ell)-1,\ldots,m_n(\ell)} \\
    & \equiv
     \sqrt{(m_j(\ell)+1)m_k(\ell)} \sket{\bm^{-1} (m_1(\ell),\ldots,m_j(\ell)+1,\ldots,m_k(\ell)-1,\ldots,m_n(\ell))}\;.
\end{align}
That is, the matrices are
\begin{align}
\label{eq:laddersymrepresentation}
   E_{j,k} : \left \{ 
   \begin{matrix}
    \left( E_{j,k}\right)_{\ell,\ell'} & \! \! \! \!= \sqrt{(m_j(\ell)+1)m_k(\ell)} & \! \! {\rm if}  \ 
    \ell^{'}=\bm^{-1} (m_1(\ell),\ldots \! ,m_j(\ell)+1\!,\ldots \!,m_k(\ell)-1 \!,\ldots 
\!,m_n(\ell)),\cr
   \left( E_{j,k}\right)_{\ell,\ell'}  & \! \! \! \! \! \! \! \! \! \! \! \! \!   \! \! \! \! \! \! \! \!=0 \ \ {\rm otherwise}.
   \end{matrix}
   \right .  
\end{align}
This automatically gives 
the Hermitian $N$-dimensional representations $S_{j,k}$ and $A_{j,k}$ of $\tilde S_{j,k}$ and $\tilde A_{j,k}$, respectively. To obtain the representation of $\tilde H_i$ we write $\tilde H_i =\frac 1 2( \tilde E_{i,i}-\tilde E_{i+1,i+1})$. The representation of 
$\tilde E_{i,i}$ is a matrix that acts on basis vectors as
\begin{align}
\nonumber
  E_{i,i} \ket{\ell} &\equiv  
E_{i,i}\ket{m_1(\ell),\ldots,m_i(\ell),\ldots,m_n(\ell)} \\
\nonumber
& = m_i(\ell) \ket{m_1(\ell),\ldots,m_i(\ell),\ldots,m_n(\ell)} \\
& \equiv m_i(\ell) \ket{\ell} \;.
\end{align}
That is,
\begin{align}
\label{eq:diagonalsymrepresentation}
    E_{i,i} : \left \{ 
   \begin{matrix}
    \left( E_{i,i}\right)_{\ell,\ell}   =  {m_i(\ell)} ,
    \cr
   \left( E_{i,i}\right)_{\ell,\ell'}  =0 \ \ {\rm if} \ \ell \ne \ell'
   \end{matrix}
   \right .  
\end{align}
for $1 \le i \le n$. 
For the case of $n=2$, these representations are all the standard irreps in terms of angular momentum operators. A reader with a physics background might have realized that the action of the matrices $E_{j,k}$
is similar to the action of corresponding creation and annihilation bosonic operators on Fock states.
This connection will be essential.
See Appendix~\ref{app:su(2)lexicographic} for  details.

While a unitary in $SU(n)$ can be expressed as an exponential of a combination of arbitrary matrices $T_a$ that represent $\mathfrak{su}(n)$, it is also simple to express it as an exponential of matrices representing the elements in the Cartan-Weyl basis.  
 
\vspace{0.2cm}
{\bf High-dimensional $SU(n)$  problem.}
Let $n \ge 2$, $M \ge 0$ be an integer, $\epsilon \ge 0$ be the error, and 
$N$ be the dimension satisfying 
$N=\binom{M+n-1}{n-1}$.
Consider the $N$-dimensional representation 
of $\mathfrak{su}(n)$ given by the Hermitian
matrices 
 $\{H_i, S_{j,k}, A_{j,k}\}$ obtained from Eqs.~\eqref{eq:laddersymrepresentation} and~\eqref{eq:diagonalsymrepresentation}.
For $1 \le i \le n-1$ and $1 \le j < k \le n$, let $\varsigma_i \in [0,4\pi)$, $\vartheta_{j,k} \in [0,4\pi)$, and $\varphi_{j,k} \in [0,4\pi)$ be a set of $n^2-1$ arbitrary phases, and define the $SU(n)$
unitary matrix 
\begin{align}
\label{eq:SU(n)problem}
    U := \exp \left(\ri \left(\sum_{i=1}^{n-1} \varsigma_i H_i + \sum_{1 \le j < k \le n} \vartheta_{j,k} S_{j,k}+ \varphi_{j,k} A_{j,k} \right)\right) \in \mathbb C^{N \times N}.
\end{align}
The goal of the high-dimensional $SU(n)$  problem is to construct a quantum circuit $V \in \mathbb C^{N' \times N'}$ that approximates $U$ within additive error at most $\epsilon$ in spectral norm, in the sense
\begin{align}
    \| (U \ket \psi) \ket 0 - V \ket \psi \ket 0\| \le \epsilon \;,
\end{align}
for all states $\ket \psi \in \mathbb C^N$.
 The quantum circuit $V$ acts on an enlarged space where $N' \ge N$, and $\ket 0$ is some zero state of an ancilla system.

\vspace{0.2cm}

We note that, 
instead of solving this problem directly, 
we solve a simplified version that regards the fast-forwarding of simpler high-dimensional $SU(n)$ unitaries. 
That is,
we can consider instead the problem of quantum circuits implementing the individual unitaries $\exp(\ri \varsigma H_i)$, $\exp(\ri \vartheta S_{j,k})$, or $\exp(\ri \varphi A_{j,k})$.
This is because the unitary in Eq.~\eqref{eq:SU(n)problem} can always be approximated by a sequence of these simpler unitaries, due to the generalized Euler angles.

More precisely, there is an efficient 
classical preprocessing that performs $\cO(n^3)$
operations (including additions and multiplications) such that, when given the phases $\varsigma_i$, $\vartheta_{j,k}$, and $\varphi_{j,k}$ that specify $U$ in Eq.~\eqref{eq:SU(n)problem}, it outputs a 
new set of angles to equivalently write $U$ as a product of these simpler exponentials. Concretely, we are able to decompose
\begin{equation}\label{eq:quadratic-decomposition}
U =  \prod_{b=1}^{n^2 - 1}\exp(\ri \vartheta_b {O}_b),
\end{equation}
for some choice of angles $\vartheta_b \in [0,4\pi]$ and where each ${O}_b \in \{ H_i, S_{j,k},A_{j,k}\}$.

The classical method uses the well-known Givens rotations and the minimal length of the sequence is also $n^2-1$, which coincides with the dimension of $\mathfrak{su}(n)$.
This can be achieved recursively in a simple manner
for our case: considering a fundamental representation
of $\mathfrak{su}(n)$, the procedure is essentially a 
Jacobi method to diagonalize the $n\times n$ unitary matrix. 
In practice, there will be small numerical errors.
That is, to achieve error $\cO(\epsilon)$
in this decomposition, we need to provide the phases and matrix entries in the input with $\cO(\log(1/\epsilon))$ bits of precision, since $n$ is constant. When considering high-dimensional representations where $\|H_i\|=\cO(N)$, $\|S_{j,k}\|=\cO(N)$, and $\|A_{j,k}\|=\cO(N)$, we will need to set $\epsilon^\prime=\cO(\epsilon/N)$
to obtain an error in the spectral norm that is $\cO(\epsilon)$. Hence, the actual complexity of this classical preprocessing to produce the decomposition of $U$ with high precision is polynomial in $n$, $\log(N)$, and $\log(1/\epsilon)$, 
without affecting the claimed speedup.
We point to Refs.~\cite{somma2005quantum,somma2019unitary} for an efficient classical approach that provides such a sequence if $n$ is constant or $n=\polylog(N)$.

\subsection{Summary of results}
\label{sec:mainresults}

Our main result is an efficient quantum circuit for the high-dimensional $SU(n)$ transform. 
\begin{theorem}
[Efficient high-dimensional $SU(n)$ circuits]
\label{thm:main}
    Let $n \ge 2$ be a constant and $M \ge 0$ integer.
    There exists a quantum circuit $V \in \mathbb C^{N' \times N'}$ that implements the $SU(n)$ unitary in Eq.~\eqref{eq:SU(n)problem}  of dimension $N=\binom{M+n-1}{n-1}$  within additive error $\epsilon >0$, with complexity $\poly(n, \log(N),\log(1/\epsilon))$. 
    The dimension $N'$ is polynomial in $N$ and $1/\epsilon$, and the input phases $\varsigma_i$, $\vartheta_{j,k}$, and $\varphi_{j,k}$ are assumed to be given within $\cO(\log(N/\epsilon))$ bits of precision.
\end{theorem}

More precisely, we will show it suffices for the dimension to be $N'=\cO((M^{2.25}/\epsilon^{3.25})^n)$
giving an overall complexity $\cO(n^2(\log (N)+\log(1/\epsilon))^3\times \log(1/\epsilon))$, which is dominated by the quantum Hermite transform (QHT), 
and the factor $n^2$ is due to the decomposition of $U$ into $\cO(n^2)$ 
terms.
See Sec.~\ref{sec:complexity} for more details.

As mentioned, we achieve this result by considering 
the oscillator representation of $\mathfrak{su}(n)$
and a recently developed QHT~\cite{jain2025efficient}. 
In general, $U \ket{\ell}=\sum_{\ell'=0}^{N-1} \alpha_{\ell,\ell'}\ket{\ell'}$ for all basis states where $\alpha_{\ell,\ell'} \in \mathbb C$. A sketch of our construction follows. 
\begin{enumerate}
    \item Apply a unitary $V_1$ that converts $\ket{\ell}$ into the $\ell$-th state by descending lexicographic order: 
    \begin{equation}
        \ket{\ell} \mapsto \ket{m_1(\ell),\ldots,m_n (\ell)}. 
    \end{equation}

    \item Apply a unitary $V_2$ built from $n$ QHTs (i.e., $V_2={\rm QHT}^{\otimes n}$) to map to the Hermite basis, 
     \begin{equation}
       \ket{m_1(\ell),\ldots,m_n(\ell)} \mapsto \sket{\psi_{m_1(\ell)},\ldots,\psi_{m_n(\ell)}}. 
    \end{equation}
    The Hermite states $\ket{\psi_m}$ belong to an enlarged space $\mathbb C^{L}$, where $L=\poly(N/\epsilon)$.

    \item Implement the corresponding unitary $\overline U$ acting on $(\mathbb C^{L})^{\otimes n}$
    obtained from the oscillator representation,
    \begin{align}
           \sket{\psi_{m_1(\ell)},\ldots,\psi_{m_n(\ell)}} \mapsto \sum_{\ell'=0}^{N-1} \alpha_{\ell,\ell'}\sket{\psi_{m_1(\ell')},\ldots,\psi_{m_n(\ell')}} \;,
    \end{align}
     where the coefficients are according to the action of $U$ 
    and $\sum_{i=1}^n m_i(\ell) = \sum_{i=1}^n m_i (\ell')=M$ in the sum. This $\bar U$ implemented as a sequence of $\cO(n^2)$ simpler unitaries.

    \item Apply the inverse of the $n$ QHT's followed by the inverse of $V_1$,
\begin{align}
    \sket{\psi_{m_1(\ell')},\ldots,\psi_{m_n(\ell')}} \mapsto \ket{m_1(\ell'),\ldots,{m_n(\ell')}} \mapsto \ket{\ell'} \;,
\end{align}
where $\ell'  \in \{0,\ldots,N-1\}$.
\end{enumerate}
\vspace{0.2cm}

The overall result is the desired map
\begin{align}
    \ket{\ell} \mapsto  
    \sum_{\ell'=0}^{N-1} \alpha_{\ell,\ell'}\ket{\ell'} = U \ket {\ell}\;.
\end{align}
The unitary $V_1$ in steps 1 and 4 can be efficiently constructed as we show in Sec.~\ref{sec:lexicographic}.
In Sec.~\ref{sec:FFdiscreteQHO} we provide a fast forwarding result
that shows how the corresponding $\overline U$ in step 3 can be decomposed into a product of diagonal operations and Fourier transforms using the oscillator representation. These operations correspond to evolutions with operators that are quadratic in the spatial coordinates or quadratic in the momentum coordinates of the bosonic systems. While the result would deliver $\overline U$ {\em exactly} in the continuum, 
where the Hermite states $\ket{\psi_m}$ are continuous eigenfunctions of the quantum harmonic oscillator $\frac 1 2 (\hat x^2 + \hat p^2)$, 
we must describe the quantum algorithm by acting on finite dimensional spaces (of qubits). This requires us to define an appropriate finite-dimensional quantum harmonic oscillator in $\mathbb C^L$, for which we show that the results in the continuum can be recovered with arbitrary accuracy by setting the dimension $L$ properly. Proving these approximations is one of the main technical contributions of this work.

Our efficient quantum circuits for high-dimensional $SU(n)$ unitaries have an immediate application to efficient and optimal quantum expanders due to Refs.~\cite{lubotzky1988ramanujan,harrow02expander,harrow2007quantum}. 
\begin{corollary}[Explicit Ramanujan quantum expanders]
\label{cor:expanders}
    For any $D = p+1$ where p is a prime congruent to 1 modulo 4, there exists explicit
    quantum circuits that implement a Ramanujan quantum expander
    of degree $D$ and dimension $N$, within error $\epsilon >0$, in time polynomial in $\log(N)$ and $\log(1/\epsilon)$. For this quantum expander,
    the second eigenvalue can get arbitrarily close to the optimal bound $\lambda \le 2 \sqrt{D-1}/D$.
\end{corollary}

To the best of our knowledge, this is
the first known construction of explicit, efficient, and optimal quantum expanders. Our construction matches the parameters of the proposals in Refs.~\cite{BST08,harrow2007quantum}, though these are, critically, non-constructive. We thus give an algorithm to make this construction explicit.
The details and proof are in Sec.~\ref{sec:applications} and Appendix~\ref{app:Ramanujanexpander}.
Ramanujan expanders are optimal in the sense that their spectral properties mimic those of a Haar random unitary, 
thereby providing the fastest mixing or scrambling rates.
While we are unable to obtain the quantum expander exactly, we note that the error $\epsilon$ can be arbitrarily small in our construction, 
which contrasts other (inefficient) constructions, where the error is fixed by the dimension
(Cf.~\cite{hastings2007random}).

Our quantum circuits also have an application
to exponential fast-forwarding the dynamics of certain quantum systems. For such systems, the Hamiltonian $H$
is a linear combination of matrices or operators that are a high-dimensional irrep of $\mathfrak{su}(n)$.
The evolution operator takes the form of Eq.~\eqref{eq:SU(n)problem}, and it can then be simulated with complexity that is only polylogarithmic in $\|H\|$ from Thm.~\ref{thm:main}, an exponential improvement over no-fast-forwarding results. 
See Sec.~\ref{sec:applications} for examples.

Our construction involves many parameters, which we summarize in the following table for easy record-keeping.
\begin{table}[h!] \label{table:parameters}
    \centering
    \begin{tabular}{l l l}
        \toprule[.5mm]
        \textbf{} & \textbf{Definition} & \textbf{Property} \\
        \midrule
        $n$ & Dimension of the group $SU(n)$ & $n \ge 2$ \\
        $N$ & Dimension of the irrep of $SU(n)$ & Exponentially large \\
        $\epsilon$ & Error in the implementation of $U$, in spectral norm & $ \epsilon >0$
        \\
        $M$ & Total number of bosons in the oscillator representation & $N=\binom{M+n-1}{n-1}$\\
        $L$ & Dimension of each discretized harmonic oscillator & $L=\poly(M/\epsilon)$
        \\
        $N'$ & Dimension of the space for the quantum circuit & $N'=\Theta(L^{n})$\\
        \bottomrule[.5mm]
    \end{tabular}
    \caption{Description of parameters. See the following sections for details.}
    \label{tab:my_table}
\end{table}

\subsection{Related work}
\label{sec:relatedwork}

We now survey a number of related works, including prior attempts to solve the high-dimensional $SU(n)$ problem, related algorithmic tasks in basis transformation and Hamiltonian simulation, and the current landscape of Ramanujan quantum expanders. Directions for future work are outlined in Sec.~\ref{sec:conclusions}.

\subsubsection{Zalka's construction}

In Ref.~\cite{zalka2004implementing}, Zalka posed the problem of implementing high-dimensional irreducible representations of $SU(2)$ on a quantum computer and sketched a possible approach. Zalka's proposal started from the observation that $SU(2)$ is a double cover of $SO(3)$ and hence implementing the irreps of $SU(2)$ is essentially the same problem as implementing the irreps of $SO(3)$, up to some minor technical details. To implement the irreps of $SO(3)$ the author proposed to exploit the fact that, for any fixed $\ell$, the spherical harmonic functions $\{Y_{\ell,m}(x,y,z): m=-\ell,\ldots,\ell\}$ transform under rotations as an irreducible representation $\rho_\ell:SO(3) \to U(2\ell+1)$, indexed by $\ell \in \{0,\frac{1}{2},1,\frac{3}{2},2,\frac{5}{2},\ldots\}$.
Here, $U(N)$ refers to unitary matrices of dimension $N \times N$.

Consider a superposition of the form
\begin{equation}
\label{eq:ketylm}
\ket{Y_{\ell,m}} = \sum_{\mathbf{x} \in \mathbb{R}^3} Y_{\ell,m}(\mathbf{x}) f(|\mathbf{x}|) \ket{\mathbf{x}},
\end{equation}
where $f(|\mathbf{x}|)$ is any envelope function yielding a normalized state (e.g., a Gaussian or a spherical shell). Then, for any $R \in SO(3)$
\begin{equation}
\label{eq:rotform}
\sum_{\mathbf{x} \in \mathbb{R}^3} Y_{\ell,m}(\mathbf{x}) f(|\mathbf{x}|) \ket{R^{-1} \mathbf{x}} = \sum_{m'} [\rho_{\ell}(R)]_{m,m'} \ket{Y_{\ell,m'}}.
\end{equation}
Consequently, if one can efficiently construct superpositions of the form \cref{eq:ketylm} using $\poly(\log (\ell))$ gates, then one can implement $\rho_\ell$ as a quantum circuit of $\mathrm{poly}(\log (\ell))$ gates as follows.
\begin{enumerate}
    \item \label{step:stateprep} Given $m$ construct $\ket{Y_{\ell},m}$ in an ancilla register:
    \begin{align}
    \ket{m} \ket{0\ldots0} \to \ket{m} \ket{Y_{\ell,m}} \;.
    \end{align}
    \item Uncompute $m$. This can be done using Kitaev's phase estimation because $\ket{Y_{\ell,m}}$ transforms under rotations about the $z$-axis as
    \begin{align}
    R_z(\phi) \ket{Y_{\ell,m}} = e^{i m \phi}.
    \end{align}
    One must take care to ensure that $R_z(\phi)$ remains reversible despite the superposition $\sum_{\mathbf{x} \in \mathbb{R}^3}$ being necessarily discretized and truncated onto finite volume. This can be achieved by decomposing the rotation into a sequence of shears as discussed in \cite{zalka2004implementing}. 
    \item Apply the given rotation $R \in SO(3)$ to $\ket{Y_{\ell,m}}$. By \cref{eq:rotform}, this yields
    \begin{align}
    \ket{Y_{\ell,m}} \to \sum_{m'} [\rho_{\ell}(R)]_{m,m'} \ket{Y_{\ell,m'}}\;.
    \end{align}
    Again, this rotation is to be implemented reversibly as a sequence of shears.
    \item Lastly, one must implement the transformation $\ket{Y_{\ell,m'}} \to \ket{m'}$, i.e., extract the index $m'$ as a bit string. This can be achieved by using phase estimation to extract $m'$ into an ancilla register, and then running the gate-by-gate inverse of the state preparation circuit from step \ref{step:stateprep} to uncompute $\ket{Y_{\ell,m'}}$.
\end{enumerate}

The key challenge is the first step, of preparing the superposition $\ket{Y_{\ell,m}}$. Zalka proposed to do this using the method of conditional rotations, that was previously introduced in \cite{Z98} (and which was further refined by Grover and Rudolph in \cite{GR02}). This can be done efficiently provided integrals of $Y_{\ell,m}^2(\mathbf{x})$ can be efficiently approximated on various finite regions of $\mathbb{R}^3$. In \cite{zalka2004implementing}, Zalka conjectured that this is indeed possible but pointed out that it was nontrivial because for exponentially large $\ell$ the integrand $Y_{\ell,m}^2$ becomes highly oscillatory, with exponentially many nodes. To our knowledge, this question has never been resolved.

Our approach is in some ways similar to Zalka's original proposal in that we consider an isometry to a space of continuous $L^2$ functions. However, we circumvent the problem of integrating $Y_{\ell,m}^2(\mathbf{x})$ by instead taking a more direct approach via a quantum Hermite transform \cite{jain2025efficient}, which avoids the evaluation of any oscillatory integrals by using the Plancharel-Rotach formula.

\subsubsection{Schur transform}
\label{Schurxform}

While the problem we consider is distinct from the quantum Schur transform, there are some related characteristics worth emphasizing.
Let $\mathcal{H}=(\mathbb{C}^d)^{\otimes n}$ be the Hilbert space of $n$ $d$-dimensional qudits.
We can act on this Hilbert space by choosing an element $u$ in the fundamental representation of the unitary group $U(d)$, which is a $d \times d$ unitary, and applying it to each qudit:
\begin{equation}
\ket{\psi} \to u^{\otimes n} \ket{\psi}\;.
\end{equation}
We can also act on this Hilbert space $\cH$ by choosing an element $\pi \in S_n$ of the symmetric group and correspondingly permuting the $n$ qudits: 
\begin{equation}
\ket{\psi_1} \otimes \ldots \otimes \ket{\psi_n} \to M_\pi \ket{\psi_1} \otimes \ldots \otimes \ket{\psi_n} = \sket{\psi_{\pi^{-1}(1)}} \otimes \ldots \otimes \sket{\psi_{\pi^{-1}(n)}},
\end{equation}
where $u^{\otimes n}$ and $M_\pi$ are reducible unitary $d^n$-dimensional representations of the unitary group $U(d)$ and the symmetric group $S_n$, respectively. These two actions on $\mathcal{H}$ commute. 

The irreducible representations of $S_n$ are in bijective correspondence with the partitions of $n$. Any partition of $n$ into at most $d$ parts indexes an irreducible representation of $U(d)$. $U(d)$ has infinitely many irreducible representations so, for any fixed $n$, these partitions only index a special subset of them. As discussed in Ref.~\cite{BCH07}, there exists a unitary change of basis $U_{\mathrm{Schur}}$ such that
\begin{equation}
\label{eq:schurdef}
U_{\mathrm{Schur}} M_\pi u^{\otimes n} U_{\mathrm{Schur}}^{-1} =
\bigoplus_\lambda \rho_\lambda(\pi) \otimes \nu_\lambda(u),
\end{equation}
where $\lambda$ ranges over all partitions of $n$ into at most $d$ parts. (If $d < n$ then not all irreducible representations of $S_n$ appear in the decomposition of $M_\pi$.)

In Refs.~\cite{H05, BCH06, BCH07} it was shown that an $\epsilon$-approximation  of $U_{\mathrm{Schur}}$ can be constructed from $\mathrm{poly}(n,d \log (1/\epsilon))$ quantum gates. The efficiency of quantum Schur transforms for large $d$ was subsequently improved in Ref.~\cite{BFG25} where, building upon Ref.~\cite{K19}, it was shown in  that ${U}_{\mathrm{Schur}}$ can be implemented using $\cO(n^4 \ \mathrm{poly}(\log d, \log (1/\epsilon)))$ gates. The implementation with the best scaling in $n$ is from Ref.~\cite{N23}, where it was shown that ${U}_{\mathrm{Schur}}$ can be implemented using $\cO(n d^4 \ \mathrm{poly}(\log(1/ \epsilon)))$ gates.

Quantum circuits for Schur transforms give as immediate corollaries quantum circuits for implementing irreducible representations of $U(d)$ of dimension higher than $d$. Specifically, from \cref{eq:schurdef} we see that by applying the circuit for $U_{\mathrm{Schur}}^{-1}$,  followed by $u^{\otimes n}$, followed by the circuit for $U_{\mathrm{Schur}}$ we get a direct sum containing all irreducible representations of $U(d)$ indexed by partitions of $n$ into at most $d$ parts. This construction also yields irreducible representations of $SU(d)$ because the irreducible representations of $U(d)$ remain irreducible when restricted to $u \in U(d)$ such that $\textrm{det}(u) = 1$. As discussed in later sections, the irreducible representations of $U(d)$ and $SU(d)$ are indexed by a list of $d$ integers, called a \emph{highest weight}. The Schur transform  does not yield an efficient way to replicate our results in this paper, because we consider representations in which the integers in the highest weight can be exponentially large, and these only arise in the Schur transform when the number of qudits $n$ is exponentially large.

\subsubsection{Hamiltonian-based methods for Lie group representations}

Using quantum Hamiltonian simulation one may attempt to implement exponentials of a representation of the Lie algebra $\mathfrak{g}$ and thereby obtain the corresponding representation of the Lie group $G$. 
Let $\rho$ denote the representation. Given an arbitrary element of $\Gamma \in G$ it is likely difficult to implement $e^{\rho(\gamma)}$ for a corresponding $\gamma$ such that $\Gamma = e^{\gamma}$ when $\rho$ has exponentially high dimension. Instead, suppose we can find some polynomial-size set of generators
$\{\gamma_1,\ldots,\gamma_m\} \in \mathfrak{g}$ such that:
\begin{itemize}
\item[] i) Given an arbitrary element $\Gamma \in G$, one can efficiently decompose it in the form 
\begin{equation}
    \Gamma = \prod_{j=1}^m e^{\gamma_j t_j}
\end{equation}
where $t_1,\ldots,t_m \in \mathbb{R}$ are coefficients efficiently computable from $G$, and
\item[] ii) Each of $e^{\gamma_1 t_1},\ldots,e^{\gamma_m t_m}$ can be implemented by polynomial-size quantum circuits.
\end{itemize}
Then, one can concatenate these quantum circuits to obtain the representation $e^{\rho(\gamma)}$ of $G$. In particular, if $\gamma_1,\ldots,\gamma_m$ are matrices of $\poly(\log(N))$ sparsity with efficiently computable matrix elements, then for each $j=1\ldots,m$ the unitary $e^{\gamma_j t_j}$ might be implemented by a quantum circuit of size $\cO(\|\gamma_j t_j\|)$ using known Hamiltonian simulation methods such as Ref.~\cite{berry2015hamiltonian}.

In Ref.~\cite{J09} this approach is used to implement exponentially-high dimensional irreducible representations of $U(n)$, $SU(n)$ and $SO(n)$. Via the formalism of Gel'fand Tsetlin diagrams, an irreducible representation of any of these Lie Groups can be specified by an $n$-tuple of integers, called the highest weight\footnote{It can happen that different $n$-tuples index the same irrep. In particular, this occurs for $SU(2)$.}. The irreducible representations of $U(n)$, $SU(n)$, and $SO(n)$, have dimension scaling as $w^{\cO(n)}$, where $w$ is the maximum magnitude of any entry in the $n$-tuple specifying the highest weight. The norms of the Hamiltonians used in \cite{J08} grow polynomially with $w$ and are independent of $n$. Thus, although the methods of \cite{J08} efficiently implement exponentially high dimensional irreps arising from fixed $w$ and polynomially large $n$, they do not efficiently implement the exponentially high dimensional irreps considered here, namely those arising from fixed $n$ and exponentially large $w$.
This is related to the fact that the $N$-dimensional irreps of the vectors of these algebras can have spectral norm that is linear in $N$, which is exponentially large in our setting.

\subsubsection{Almost Ramanujan quantum expanders}

In Ref.~\cite{jeronimo2025almost} the authors present an efficient algorithm to produce almost Ramanujan expanders from any expander. Their approach extends to the construction of almost Ramanujan {\em quantum} expanders that are also efficient and explicit as in our case. Their results are formulated as a bound on the degree $D$ in terms of the ``expansion parameter'' $\lambda$, which is the second largest eigenvalue. More precisely, they achieve $D=\cO(1/\lambda^{2 + o_\lambda(1)})$, where 
$o_\lambda(1)$ scales as $1/\sqrt{\log(1/\lambda)}$ and is asymptotically zero as $\lambda \rightarrow 0$. However, for small, constant $D$, the second eigenvalue is $\lambda>0$ and hence the dependence of $D$ with $\lambda$ is worse than $1/\lambda^2$ and suboptimal. 
In contrast, our construction does not have this additional dependence and satisfies the Ramanujan condition between $D$ and $\lambda$ even for small $D$.

\section{Oscillator representation of $SU(n)$}
\label{sec:oscillatorrep}

In this section, we present essential concepts regarding the oscillator, or Jordan-Schwinger, representation, which allows us to realize $\mathfrak{su}(n)$ using $n$ quantum harmonic oscillators. This will form the core of our fast-forwarding method, enabling us to factorize these rotations into exponentials of quadratic monomials in position and momentum.
Later, in Sec.~\ref{sec:finiteQHO}, we will discretize the representation presented here to obtain bounds on the error of our procedure.

\subsection{Jordan-Schwinger mapping}
As mentioned, the $N$-dimensional representation
of  $SU(n)$ of interest can be obtained by exponentiating
the corresponding $N$-dimensional matrices from the irreducible representation of the algebra $\mathfrak{su}(n)$. 
One way of obtaining the latter is via the oscillator representation, also known as the Jordan-Schwinger mapping. 
In essence, we will use the fact that the group $SU(n)$
is a subgroup of the symplectic group $Sp(2n, \mathbb R)$,
and the latter can be generated by certain bosonic operators.

To this end, we consider $n$ bosonic modes, or quantum harmonic oscillators, where the total number of bosons is set to some $M \ge 0$. Let $a^\dagger_j$ and $a^{}_j$ be the creation and annihilation operators of the $j^{\rm th}$ bosonic mode. Then we  define $\hat E_{j,k}:=a^\dagger_j a^{}_k$ and $\hat H_i =\frac 1 2 ( a^\dagger_i a^{}_i - a^\dagger_{i+1} a^{}_{i+1})$, where $1 \le j < k \le n$ and $1 \le i \le n-1$. Accordingly, we define $\hat S_{j,k}:=\frac 1 2 (\hat E_{j,k} + (\hat E_{j,k})^\dagger)$ and 
$\hat A_{j,k}:=\frac {\ri} 2 (\hat E_{j,k} - (\hat E_{j,k})^\dagger)$.
The canonical commutation relations of bosonic operators are $[a^{}_j,a^\dagger_k]=\delta_{j,k}$, where $\delta_{j,k}$ is the Kronecker delta. These in turn imply that such operators also satisfy the same commutation relations in the abstract Lie algebra $\mathfrak{su}(n)$, and hence they also span $\mathfrak{su}(n)$. See Appendix~\ref{app:schwingerrep} for more details.

In the Jordan-Schwinger representation, an irreducible representation of $\mathfrak{su}(n)$ is obtained by considering the subspace spanned by the bosonic states of definite particle number $M$. This is the subspace spanned by Fock states
\begin{align}
    \cB_{n,M} :=\left \{  \ket{m_1,m_2,\ldots,m_n} : \sum_{j=1}^n m_j = M \right \} \;.
\end{align}
The dimension of this subspace is $N=\binom{M+n-1}{n-1}$.
For now and for simplicity, $\ket{m_1,m_2,\ldots,m_n}$ represents the state with $m_1$ bosons in the first site, $m_2$ bosons in the second site, and so on.
The representation is obtained using the known properties
of the bosonic operators, namely
\begin{align}
\label{eq:creation}
a^\dagger_j \ket{m_1,\ldots,m_j, \ldots, m_n}& = \sqrt{m_j+1}\ket{m_1,\ldots,m_j+1, \ldots, m_n} \;,\\
\label{eq:annihilation}
a^{}_j \ket{m_1,\ldots,m_j, \ldots, m_n}& = \sqrt{m_j}\ket{m_1,\ldots,m_j-1, \ldots, m_n} \;.
\end{align}

Consider now the standard descending lexicographic ordering for the basis $\cB_{n,M}$, given by strings of size $n$. For example, 
 $\cB_{3,2}=\{\ket{3,0,0},\ket{2,1,0},\ket{2,0,1},\ket{0,3,0},\ket{0,2,1},\ket{0,0,3}\}$. More generally, labeling these states as $\{\ket{\ell}:0 \le \ell \le N-1\}$, the ordering gives a one-to-one map such that $\ket{\ell} \equiv \ket{m_1(\ell),\ldots,m_n(\ell)}$. Then, Eqs.~\eqref{eq:creation} and~\eqref{eq:annihilation} imply
 \begin{align}
 \nonumber
     \hat E_{j,k} \ket{\ell} &\equiv  \hat E_{j,k}
     \ket{m_1(\ell),\ldots,m_j(\ell),\ldots,m_k(\ell),\ldots,m_n(\ell)} \\
     \nonumber
     &= \sqrt{(m_j(\ell)+1)m_k(\ell)}\ket{m_1(\ell),\ldots,m_j(\ell)+1,\ldots,m_k(\ell)-1,\ldots,m_n(\ell)} \\
     & \equiv\sqrt{(m_j(\ell)+1)m_k(\ell)} \ket{\ell'}\;,
 \end{align}
 where $\ell'$  is such that $m_1(\ell')=m_1(\ell),\ldots,m_j(\ell')=m_j(\ell)+1,\ldots,m_k(\ell')=m_k(\ell)-1,\ldots,m_n(\ell')=m_n(\ell)$. Also,
 \begin{align}
 \nonumber
     \hat H_i \ket{\ell} &\equiv    \hat H_i \ket{m_1(\ell),\ldots,m_i(\ell),m_{i+1}(\ell),\ldots,m_n(\ell)} \\
     \nonumber
     &= (m_i(\ell)-m_{i+1}(\ell) )\ket{m_1(\ell),\ldots,m_i(\ell),m_{i+1}(\ell),\ldots,m_n(\ell)} \\
     & \equiv(m_i(\ell)-m_{i+1}(\ell) )\ket{\ell}\;.
 \end{align}
 The action of these operators on the basis $\cB_{n,M}$
 immediately gives an $N$-dimensional representation for them, where $\hat H_i$, $\hat S_{j,k}$, and $\hat A_{j,k}$ are represented by the Hermitian matrices $H_i$, $S_{j,k}$, and $A_{j,k}$, respectively. Simple inspection shows that these matrices are the ones obtained from Eqs. \eqref{eq:laddersymrepresentation} and~~\eqref{eq:diagonalsymrepresentation}. 
 That is, the Jordan-Schwinger representation provides
 the totally symmetric representation of $SU(n)$ of interest.

 A viable path towards implementing the high-dimensional $SU(n)$ unitary is then by implementing the corresponding bosonic operators over the subspace spanned by the Fock states.  These Fock states are the eigenstates (eigenfunctions) of the quantum harmonic oscillator. That is, consider the Hamiltonian
\begin{align}
    \hat H := \frac 1 2 (\hat x^2 + \hat p^2) \;,
\end{align}
where $\hat x$ and $\hat p$ are   the position and momentum operators, acting on $L^2$ functions as $\hat x f(x)=xf(x)$
and $\hat p f(x)=-\ri \partial_x f(x)$. The eigenfunctions of $\hat H$ are the well-known (physicist's) Hermite functions:
\begin{align}
\label{eq:Hermitefunctions}
    \psi_m(x) : = \frac 1 {\sqrt {2^m m! \sqrt \pi}}e^{-x^2/2} H_m(x) \implies \hat H \psi_m(x) = (m+1/2) \psi_m(x) \;,
\end{align}
for all integer $m \ge 0$.
Here, $H_m(x)$ is the $m^{\rm th}$ Hermite polynomial that satisfies
\begin{align}
    H_m(x) =(-1)^m e^{x^2} \frac{\rd^m}{\rd x^m}e^{-x^2}\;.
\end{align}
Then, in our notation for $\cB_{n,M}$, we can express
\begin{align}
    \ket{m_1,\ldots,m_n} \equiv \left( \int \rd x_1 \;\psi_{m_1}(x_1)\ket{x_1}\right) \otimes \cdots \otimes \left( \int \rd x_n \; \psi_{m_n}(x_n)\ket{x_n}\right)\;,
\end{align}
where $\ket x$ represents the eigenstate of $\hat x$, that is, $\hat x \ket x = x \ket x$, and $x \in \mathbb R$.

In the following, we will represent the eigenstates of the quantum harmonic oscillator as $\tshostate{m}:=
\int \rd x \;\psi_{m} (x)\ket{x}$, so that $ \ket{m_1,\ldots,m_n} \rightarrow \ket{\psi_{m_1}^{\rm c},\ldots,\psi_{m_n}^{\rm c}}:=\tshostate{m_1} \otimes \ldots \otimes  \tshostate{m_n}$ also represent the Fock states. (The superscript $c$  denotes the states in the continuum.)
This is useful to make the distinction with the standard computational-basis states used in the description of our quantum algorithm (i.e., the unitary $U$) and the states appearing in the discrete quantum harmonic oscillator.

Next we show how to efficiently implement $U$ by showing how to implement the corresponding $\hat U$ in the 
 $N$-dimensional subspace spanned by Fock states.

\subsection{Fast-forwarding $SU(n)$ in the oscillator representation}
\label{sec:fast-forwarding overview}

In the oscillator representation, any $SU(n)$
unitary transformation $\hat U$ can be written as a product of simple unitaries, each involving the exponential 
of a single basis operator $\hat H_i$, $\hat S_{j,k}$, or $\hat A_{j,k}$ only. 
We now show a way to fast-forward each of these operations. To this end, it is convenient to write the creation and annihilation operators in terms of the position and momentum operators. That is, for $1 \le j \le n$, we define
\begin{align}
   \hat x_j := \frac 1 {\sqrt 2} (a^\dagger_j + a^{}_j), \; \hat p_j := \frac{\ri}{\sqrt 2}(a^\dagger_j - a^{}_j). 
\end{align}
The commutation relations for $a^\dagger_j$ and $a^{}_j$ imply
\begin{align}
    [\hat x_j, \hat p_k]=\ri \delta_{j,k}, \; [\hat x_j,\hat x_k]=[\hat p_j,\hat p_k]=0\;.
\end{align}
The corresponding Lie algebra of position and momentum operators is known as the Heisenberg algebra, and the Lie algebra of quadratic position and momentum operators is ${\mathfrak {sp}}(2n,\mathbb R)$, which contains ${\mathfrak{su}}(n)$. 
It follows that
\begin{align}
\label{eq:HiJSrep}
    \hat H_i &= \frac 1 2 \left((\hat x_i)^2 + (\hat p_i)^2 - (\hat x_{i+1})^2- (\hat p_{i+1})^2 \right) \; , \\
    \label{eq:SjkJSrep}
    \hat S_{j,k} &=\frac 1 2 \left(\hat x_j \hat x_k + \hat p_j \hat p_k \right)\; , \\
     \label{eq:AjkJSrep}
    \hat A_{j,k} &=\frac 1 2 \left(\hat p_j \hat x_k - \hat x_j \hat p_k \right)\;,
\end{align}
where $1 \le i \le n-1$ and $1 \le j<k\le n$.

In the real-space basis where $\hat x_j$ is diagonal (i.e., the basis spanned by $\ket x$), we assume it is ``easy'' to implement the corresponding unitaries obtained via exponentiation of position operators, since
these unitaries are diagonal. A similar result would follow for exponentials of momentum operators, since position and momentum are related via conjugation with the Fourier transform, and this can also be implemented ``easily'': $\momentum=\hat F^{-1} \position \hat F^{}$, where $F$ is the Fourier transform in $\mathbb R$. However, exponentials of linear combinations involving both position and momentum are a priori not easy to implement. Nevertheless, if one is able to decompose such exponentials as a sequence of exponentials involving position or momentum operators only, then the corresponding exponentials become easily implementable.

To show this is possible, 
we build upon the following known result~\cite{QA07,jain2025efficient}.
\begin{claim}
\label{claim:main}
Let $K_1, K_2, K_3$ be three operators in a Lie algebra obeying the commutation relations
\begin{align}
[K_1, K_2] = -2K_3, \quad
[K_3, K_1] = 2K_1, \quad
[K_3, K_2] = -2K_2.
\end{align}
Then for all $t \in \mathbb{R}$,
\begin{align}
e^{(K_1+K_2)t} = e^{\alpha K_2} e^{\beta K_1} e^{\alpha K_2},
\end{align}
where
\begin{align}
\alpha = \frac{\tan(t/\sqrt{2})}{\sqrt{2}}, 
\qquad 
\beta = \frac{\sin(\sqrt{2} t)}{\sqrt{2}}.
\end{align}
\end{claim}

Up to constant factors, the Lie algebra spanned by $K_1$, $K_2$, and $K_3$
is the well known $\mathfrak{sp}(2,\mathbb R)$
Lie algebra of the quantum harmonic oscillator. 
This claim is what allowed for the fast forwarding result in Ref.~\cite{jain2025efficient}, but we now use for $n$ coupled quantum harmonic oscillators.

\begin{lemma}
\label{lem:FFQHOs}
    Let $\varsigma, \ \vartheta, \ \varphi \in [-\pi/2,\pi/2]$, $n \ge 2$, and consider the operators $\hat H_i$, $\hat S_{j,k}$, and $\hat A_{j,k}$ of Eqs.~\eqref{eq:HiJSrep},~\eqref{eq:SjkJSrep}, and~\eqref{eq:AjkJSrep}, where $1 \le i \le n-1$ and $1 \le j < k \le n$. Then,
    \begin{align}
        e^{\ri \varsigma \hat {H}_i}& = e^{\ri \varsigma_1   ({\hat p}_i)^2} 
        e^{\ri \varsigma_2  ({\hat x}_i)^2}
        e^{\ri \varsigma_1   ({\hat p}_i)^2} 
        e^{\ri \varsigma_1   ({\hat p}_{i+1})^2} 
        e^{\ri \varsigma_2  ({\hat x}_{i+1})^2}
        e^{\ri \varsigma_1   ({\hat p}_{i+1})^2}
          \;, \\
         e^{\ri \vartheta \hat {S}_{j,k}}& =
         e^{\ri \vartheta_1 \hat p_j \hat p_k}
        e^{\ri \vartheta_2 \hat x_j \hat x_k}
         e^{\ri \vartheta_1 \hat p_j \hat p_k}
         \;, \\
         e^{\ri \varphi \hat {A}_{j,k}}& =
         e^{-\ri \varphi_1   {\hat x}_j \hat p_k}  e^{\ri \varphi_2   {\hat p}_j \hat x_k}  e^{-\ri \varphi_1   {\hat x}_j \hat p_k} \;,
    \end{align}
    where $\varsigma_1= \tan(\varsigma/2)/ 2$,
    $\varsigma_2= \sin(\varsigma)/ 2$, $\vartheta_1=\tan(\vartheta/4)$, $\vartheta_2=\sin(\vartheta/2)$, 
    $\varphi_1=\tan(\varphi/4)$, and $\varphi_2=\sin(\varphi/2)$.
\end{lemma}

\begin{proof}
    For the first result note that the operators for the $i^{\rm th}$ mode commute with that of 
    the $(i+1)^{\rm th}$ mode and hence $\exp(\ri \varsigma \hat H_i)=\exp(\ri \varsigma (\hat x_i^2 + \hat p_i^2)/2)\exp(\ri \varsigma (\hat x_{i+1}^2 + \hat p_{i+1}^2)/2)$ factorizes. Then
    we can set $K_1=\ri (\hat x_j)^2/\sqrt 2$,  $K_2=\ri (\hat p_j)^2/\sqrt 2$, and $K_3=\ri (\hat x_i \hat p_i + \hat p_i \hat x_i)/2$, that satisfy the commutation relations in Claim~\ref{claim:main}, which would imply
    \begin{align}
   e^{\ri \varsigma (\hat x_i^2 + \hat p_i^2)/2}  =
   e^{ (\varsigma/\sqrt 2) (K_1+K_2)} =   e^{\ri \varsigma_1   ({\hat p}_i)^2} 
        e^{\ri \varsigma_2  ({\hat x}_i)^2}
        e^{\ri \varsigma_1   ({\hat p}_i)^2} \;.
    \end{align}
    Next we apply the same analysis using the $(i+1)^{\rm th}$ mode and obtain the expression.

   For the second result we consider $K_1=\ri \sqrt 2 (\hat x_j \hat x_k)$, $K_2=\ri \sqrt 2 (\hat p_j \hat p_k)$, and $K_3=\ri (\hat x_j \hat p_j + \hat p_k \hat x_k)$. These also satisfy the commutation relations in Claim~\ref{claim:main}, implying
   \begin{align}
       e^{\ri \vartheta \hat S_{j,k}}= e^{ (\vartheta/(2\sqrt 2)) (K_1+K_2)} = 
       e^{\ri \vartheta_1 \hat p_j \hat p_k}
        e^{\ri \vartheta_2 \hat x_j \hat x_k}
         e^{\ri \vartheta_1 \hat p_j \hat p_k}\;.
   \end{align}
Similarly, for the third result, we can choose $K_1 = \ri \sqrt 2 \hat p_j \hat x_k$, 
   $K_2 = -\ri \sqrt 2 \hat x_j \hat p_k$, and 
   $K_3 = \ri(\hat x_j \hat p_j - \hat x_k \hat p_k)$. Then,
   \begin{align}
  e^{\ri \varphi \hat A_{j,k}}= e^{(\varphi/(2\sqrt 2)) (K_1+K_2)} = 
       e^{-\ri \varphi_1 \hat x_j \hat p_k}
        e^{\ri \varphi_2 \hat p_j \hat x_k}
          e^{-\ri \varphi_1 \hat x_j \hat p_k}.  
   \end{align}
\end{proof}

Hence, we are able to factorize the exponentials of the operators in the desired way. Furthermore, since any unitary $U$ in $SU(n)$ can be factorized as a sequence of simpler exponentials, then the corresponding $\hat U$ acting in the oscillator representation can be written as
\begin{align}
\label{eq:Uinoscillatorsrep}
    U \mapsto \hat U=\prod_{b} \exp(\ri \vartheta_b \hat O_b)\;,
\end{align}
where the $\hat O_b$'s are now quadratic monomials in position and momentum, that is, 
\begin{align}
    \hat O_b \in \{ \hat x_j \hat x_k , \hat x_j \hat p_k, \hat p_j \hat p_k\}.
\end{align}
Moreover, the number of terms in the product is bounded as $r=\cO(n^2-1)$ since Lemma~\ref{lem:FFQHOs}
shows that each exponential is a product of at most
six exponentials.
Also, the phases $\vartheta_b$
can  be bounded by a constant (e.g., $|\vartheta_b|\le 1$), since the phases in 
Lemma~\ref{lem:FFQHOs} are also bounded. 
More generally, if the phases where to be larger than say $\pi/2$, we can simply split the exponential into a product of exponentials involving smaller phases; for example, $e^{\ri \varsigma \hat H_i}= 
(e^{\ri \varsigma \hat H_i/T})^T$. 

To obtain the desired circuit for $U$ when acting on basis states, we would like to prove that we can discretize and assume the space of the harmonic oscillator is finite dimensional without incurring in large errors, while still preserving the fast-forwarding feature, and that we can prepare the corresponding Hermite states efficiently.
The last step is precisely what the quantum Hermite transform (QHT) in Ref.~\cite{jain2025efficient} achieves. We will address these items next. 

\section{Discrete quantum harmonic oscillator}
\label{sec:finiteQHO}

Equation~\eqref{eq:Uinoscillatorsrep} is an operator that acts in the continuous space of $n$ quantum harmonic oscillators.
Our quantum algorithm runs on a system of qubits, and hence we will need to perform these mappings within finite dimensional spaces. To this end, we consider a discrete, finite-dimensional version of the quantum harmonic oscillator that allows us to recover the results of Sec.~\ref{sec:fast-forwarding overview} in the continuum. This is the same system studied in Refs.~\cite{somma2016,jain2025efficient}.

\paragraph{A single discrete QHO.} 
The discrete QHO is modeled by the $L$-dimensional Hamiltonian \begin{align}
\label{eq:discretehamiltonian}
\discretehamiltonian = \frac{1}{2}\left(\discreteposition^2 + \discretemomentum^2\right) \;,
\end{align}
where 
\begin{align}
\discreteposition := \sqrt{\frac{2\pi}{L}}\begin{pmatrix}
    -\frac{L}{2} & 0 & ... & 0 \\
    0 & -\frac{L}{2} + 1 & ... & 0 \\
    \vdots & \vdots & \ddots & \vdots \\
    0 & 0 & ... & \frac{L}{2}-1
\end{pmatrix} \;
\end{align}
is the discrete position operator and $\discretemomentum := F^{-1}\discreteposition F$ 
the discrete momentum operator, and $F$ is the centered $L$-dimensional discrete Fourier transform.
The dimension $L$ needs to be set according to the precision requirements, since it determines the discretization size $\sqrt{ 2\pi/{L}}$ and the quality of the approximation.
For arbitrary integer $m \ge 0$, we define 
the discrete Hermite states  via
\begin{align}\label{eq:hermite-state}
\dshostate{m} := \left( \frac{2\pi} L\right)^{1/4} \sum_{j=-L/2}^{L/2-1} \psi_m(x_j)\ket{j} \;,
\end{align}
 where $\psi_m(x)$ is the $m^{\rm th}$ Hermite function
 of Eq.~\eqref{eq:Hermitefunctions} and $x_j:= j  \sqrt{\frac{2\pi}{L}}$ is a point in the discretized space.
 In a ``low-energy'' subspace where $m \le cL$ for some constant $c<1$, these can be shown to be approximately the eigenstates of $\discretehamiltonian$,
 and they recover many properties of the states $\tshostate{m}$ in the continuum.
 
 That is, Somma
(Ref.~\cite{somma2016}) showed that there exist constants $c,\gamma\in(0,1)$ such that, 
for all $L>0$ sufficiently large and all $0 \le m,m' \leq c L$,
\begin{align}
\label{eq:singleQHOerrorbounds}
\left|\bradshostate{m'}\overline x^{a} \overline p^{b} \dshostate{m} - \bra{\psi_{m'}^{\rm c}}\hat x^{a}\hat p^{b}\ket{\psi_{m}^{\rm c}}\right|
\le \exp(-\gamma L).
\end{align}
Here, $a$ and $b$ are constant and non-negative integers (e.g., $a,b \le 4$). Furthermore, 
for $0 \le m \le cL$,
\begin{align}
    \| (\discretehamiltonian - (m+1/2)) \ket{\psi_m}\| \le \exp(-\gamma L) \;,
\end{align}
and there are similar approximations for other quantities.
The last equation implies these Hermite states are almost eigenstates of $\discretehamiltonian$. 
In the case of the $SU(n)$ transformation, and since the relevant space is of dimension $N=\binom{M+n-1}{n-1}$, we will choose $L$
large enough so that $c L \gg M$\footnote{Numerical simulations show that the constant satisfies $c \ge 3/4$.}. The fact that these errors are exponentially small in $L$   will be advantageous.

\paragraph{Many discrete QHOs.}
 In contrast with Ref.~\cite{jain2025efficient}, here we need to invoke $n \ge 2$ discrete quantum harmonic oscillators. The Hilbert space is then $(\mathbb C^L)^{\otimes n}$, of dimension $L^n$. For each system $i\in\{1,\ldots,n\}$,
 we define the corresponding operators as
\begin{align}
\discreteposition_i := \one^{\otimes (i-1)}\otimes \discreteposition \otimes \one^{\otimes(n-i)},\qquad
\discretemomentum_i := \one^{\otimes (i-1)}\otimes \discretemomentum\otimes \one^{\otimes(n-i)}.
\end{align}
We also consider the tensor product of Hermite states
in $(\mathbb C^L)^{\otimes n}$:
\begin{align}
\sket{\psi_{m_1},\ldots,\psi_{m_n}}:=    \dshostate{m_1}\otimes \ldots \otimes  \dshostate{m_n}
\end{align}
The error bounds in Eqs.~\eqref{eq:singleQHOerrorbounds}
can be then directly carried to the case of $n$ quantum harmonic oscillators, where we can now consider products of the $\discreteposition_i$'s and $\discretemomentum_i$'s.
In the next section, we will use these bounds to show that the fast-forwarding identities of
Sec.~\ref{sec:fast-forwarding overview}, which are exact in the continuum, remain valid up to exponentially small error on a desired low energy subspace of the $n$ harmonic oscillators, if the local dimension $L$ is chosen appropriately large.

\subsection{Discretization errors}
\label{sec:FFdiscreteQHO}

Thus far we showed that the unitary
in Step 3 of Sec.~\ref{sec:mainresults} could have been implemented exactly via a simple product of exponentials of quadratic monomials in $\hat x_j$ and $\hat p_k$, if we were allowed to work in the continuum. This is the unitary $\hat U$ in Eq.~\eqref{eq:Uinoscillatorsrep}.
Now we show that it is possible to discretize the space and implement that unitary within arbitrary, exponentially small errors, which we call $\overline U$. To this end, we will replace $\hat x_j \to \discreteposition_j$ and $\hat p_j \to \discretemomentum_j$ in each exponential of Lemma~\ref{lem:FFQHOs}, giving us $\overline U$. 
We assume $n \ge 2$ to be a constant, however the results might be adapted for the case $n =\polylog(N)$.

\begin{lemma}\label{lem:matrix-element-closeness}
    Let $\hat U$ be the unitary of Eq.~\eqref{eq:Uinoscillatorsrep} that implements the $SU(n)$ operation in the oscillator representation,
    where the number of terms in the product is $r \ge 1$.
    Consider its action on an arbitrary Fock state as
    \begin{align}
        \hat U \sket{\psi_{m_1}^{\rm c},\ldots,\psi_{m_n}^{\rm c}} = \sum_{m'_1,\ldots,m'_n} \beta^{m_1,\ldots,m_n}_{m'_1,\ldots,m'_n}
       \sket{\psi_{m'_1}^{\rm c},\ldots,\psi_{m'_n}^{\rm c}} \;,
    \end{align}
    where the sums are such that $\sum_i m_i = \sum_i m'_i = M$. Consider the corresponding finite-dimensional unitary $\overline U$,
    obtained by replacing 
     $\hat x_j \to \discreteposition_j$ and $\hat p_j \to \discretemomentum_j$ in Eq.~\eqref{eq:Uinoscillatorsrep}.
     Then, if the dimension of each discrete quantum harmonic oscillator satisfies $L \ge c' 2^r M$,
     for some constant $c'>1$,
     we obtain
     \begin{align}
       \| \overline U  \sket{\psi_{m_1},\ldots,\psi_{m_n}} -\sum_{m'_1,\ldots,m'_n} \beta^{m_1,\ldots,m_n}_{m'_1,\ldots,m'_n}
       \sket{\psi_{m'_1},\ldots,\psi_{m'_n}}\| \le  \exp(-\overline \gamma L) \;,
    \end{align}
    where $\overline \gamma >0$ is a positive constant.  
\end{lemma}

This Lemma essentially shows that Step 3 in Sec.~\ref{sec:mainresults} can be implemented using simple unitaries, which are exponentials of quadratic monomials of the discrete position and momentum operators, within exponentially small error. Each such exponential is either a diagonal unitary, or a diagonal unitary conjugated by the centered Fourier transform, and hence they can be fast-forwarded using standard techniques.

The error is exponentially small in $L$. At the same time, we need to choose $L =\Omega(2^r M)$ for the result to work, where $r =\cO(n^2)$. If $n$ is constant,
as we assumed, this is $L =\Omega(M)$. For error $\epsilon$, then setting some
$L = \cO(M + \log(1/\epsilon))$ suffices. 
However, due to the requirements of the QHT in Ref.~\cite{jain2025efficient}, we will ultimately need to choose a much larger $L = \cO(M^{2.25}/\epsilon^{3.25})$. This will imply that the error in the transformation due to discretization in Lemma~\ref{lem:matrix-element-closeness} could be much smaller than $\epsilon$.

We split the proof Lemma~\ref{lem:matrix-element-closeness} using the following three results. 

\begin{lemma}\label{lem:simple-exponential-bound-x}
    Consider the simpler unitary $\hat W= \exp(\ri \vartheta \hat x_j \hat x_k)$ acting on two quantum harmonic oscillators, and the corresponding discretization $\overline W = \exp(\ri \vartheta \overline x_j \overline x_k)$. Then, 
    if $\hat W \sket{\psi_{m_j}^{\rm c},\psi_{m_k}^{\rm c}}= \sum_{m_j^\prime,m_k^\prime=0}^\infty
    \alpha_{m_j^\prime m_k^\prime}\sket{\psi_{m_j^\prime }^{\rm c},\psi_{m_k^\prime}^{\rm c}}$, we obtain
    \begin{align}
        \overline W \sket{\psi_{m_j},\psi_{m_k}}= \sum_{m_j^\prime,m_k^\prime=0}^\infty
    \alpha_{m_j^\prime m_k^\prime}\sket{\psi_{m_j},\psi_{m_k}} \;.
    \end{align}
\end{lemma}

\begin{proof}
We will assume without loss of generality that $j = 1$ and $k = 2$. The proof for all other cases follows in the exact same manner, even when $j = k$.
First, we observe that the discrete Hermite states satisfy a similar recurrence than that of the Hermite functions. That is
\begin{equation}
\label{eq:cont-positionrecurrence}
\position\tshostate{m}
=\sqrt{\frac{m}{2}}\tshostate{m-1}+\sqrt{\frac{m+1}{2}}\tshostate{m+1}
\implies \discreteposition\dshostate{m}
=\sqrt{\frac{m}{2}}\dshostate{m-1}+\sqrt{\frac{m+1}{2}}\dshostate{m+1} \;.
\end{equation}
Then, we  write $\hat{W}$ and $\overline{W}$ as (convergent) Taylor series:
\begin{align}
\hat{W} = \sum_{k=0}^{\infty} \frac{(\ri \theta)^k (\hat{x}_1 \hat{x}_2)^k}{k!} = \sum_{k=0}^{\infty} \frac{(\ri \theta)^k \hat{x}_1^k \hat{x}_2^k}{k!} , \ 
\overline{W} = \sum_{k=0}^{\infty} \frac{(\ri \theta)^k \overline{x}_1^k \overline{x}_2^k}{k!}.
\end{align}
For each $k$, we write without loss of generality
\begin{align}
\hat x_1^k \hat x_2^k \sket{\psi_{m_1}^{\rm c},\psi_{m_2}^{\rm c}} = \sum_{m_1^\prime, m_2^\prime} \gamma_{m_1^\prime, m_2^\prime}^{(k)} \sket{\psi_{m'_1}^{\rm c},\psi_{m'_2}^{\rm c}},
\end{align}
and the coefficients follow from the recurrence in Eq.~\eqref{eq:cont-positionrecurrence}.
This recurrence also implies
\begin{align}
\overline x_1^k \overline x_2^k \sket{\psi_{m_1},\psi_{m_2}} = \sum_{m_1^\prime, m_2^\prime} \gamma_{m_1^\prime, m_2^\prime}^{(k)} \sket{\psi_{m'_1},\psi_{m'_2}},
\end{align}
with exactly the same coefficients.
Thus, combining terms over all $k$, we can obtain $\alpha_{m_1^\prime, m_2^\prime} = \sum_{k=0}^\infty \gamma_{m_1^\prime, m_2^\prime}^{(k)}$ in both the discrete and continuum cases,
proving the desired result:
\begin{align}
        \overline W \sket{\psi_{m_1},\psi_{m_2}}= \sum_{m'_1,m'_2=0}^\infty
    \alpha_{m'_1m'_2}\sket{\psi_{m_1},\psi_{m_2}} \;.
    \end{align}
\end{proof}

Note that, since $\dshostate{m}$ are defined in the space of dimension $L$, they eventually become linearly dependent in the above sum. Nevertheless, we can still 
write the action of $\overline W$ in that form, and also note that later we will cutoff the sums at proper values for $m$. The result also holds for other unitaries in Lemma~\ref{lem:FFQHOs} that are exponentials of $\position_i^2$, if we replace them by exponentials of $\discreteposition_i^2$, since the recurrence relation applies for the position operators. 
Next we prove an analogous result for momentum operators.

\begin{lemma}\label{lem:simple-exponential-bound-p}
    Consider the simpler unitary $\hat W= \exp(\ri \vartheta \hat p_1 \hat p_2)$ acting on two quantum harmonic oscillators, and the corresponding one $\overline W$ obtained via discretization where $\momentum_i$ is replaced by $\discretemomentum_i$. Let the phase be $\vartheta \in [-1/2e, 1/2e]$. Then, 
    if $\hat W \sket{\psi^{\rm c}_{m_1},\psi^{\rm c}_{m_k}}= \sum_{m_j^\prime ,m_k^\prime=0}^\infty
    \alpha_{m_j^\prime m_k^\prime}\sket{\psi^{\rm c}_{m_j^\prime },\psi^{\rm c}_{m_k^\prime}}$, we obtain
    \begin{align}
       \| \overline W \sket{\psi_{m_j},\psi_{m_k}} - \sum_{m_j^\prime ,m_k^\prime=0}^\infty
    \alpha_{m_j^\prime m_k^\prime}\sket{\psi_{m_j^\prime },\psi_{m_k^\prime}} \| 
   \le  \exp(-\tilde \gamma L)
    \end{align}
    for all $m_j, m_k \le c' L$,
     where $c' <1$ and $\tilde \gamma >0$ are positive constants.
\end{lemma}
\begin{proof}
Again, we assume without loss of generality that $j = 1, k = 2$.
We would like to follow the same proof as in Lemma~\ref{lem:simple-exponential-bound-x}, but unfortunately the recurrence relation does not apply to the discrete momentum operators. Hence, 
we make use of the $L$-dimensional centered Fourier transform $F$ that
maps $\discreteposition\mapsto\discretemomentum\mapsto-\discreteposition$, and write
\begin{align}
   \overline W \sket{\psi_{m_1},\psi_{m_2}} =(F^{-1} \otimes F^{-1}) \overline V (F \otimes F) \sket{\psi_{m_1},\psi_{m_2}} \; , \ \overline V:=\exp(\ri \vartheta \discreteposition_1 \discreteposition_2)\;.
\end{align}
Since the Hermite states $\dshostate{m}$ are not
exactly the eigenstates of $\discretehamiltonian$, they are not eigenstates of $F$, which contrasts the case in the continuum. However, they can be proven to be almost eigenstates and Ref.~\cite{somma2016} shows
\begin{align}
\label{eq:Fouriereigenvalue}
    \|F \dshostate{m}-\ri^m \dshostate{m}\| \le \exp(-\gamma L)
\end{align}
if $m \le cL$, where $\gamma>0$ and $c<1$ are positive constants.
Using a Taylor series expansion for $\overline V$,
it follows that
\begin{align}
\nonumber
    \overline V (F \otimes F) \sket{\psi_{m_1},\psi_{m_2}} &= \ri^{m_1+m_2} \overline V \sket{\psi_{m_1},\psi_{m_2}}+ \cO(\exp(-\gamma L)) \\
    \nonumber
    & = \ri^{m_1+m_2} \sum_{k=0}^\infty \frac{(\ri \vartheta)^k (\discreteposition_1 \discreteposition_2)^k}{k!}\sket{\psi_{m_1},\psi_{m_2}}+ \cO(\exp(-\gamma L)) 
    \\
    \label{eq:VbarHermite}
    &=\sum_{m'_1,m'_2=0}^\infty \gamma_{m'_1,m'_2} \sket{\psi_{m'_1},\psi_{m'_2}} + \cO(\exp(-\gamma L))\;,
\end{align}
where the last line follows from Eq.~\eqref{eq:cont-positionrecurrence}.
The coefficients $\gamma_{m'_1,m'_2}$ match those of the continuum, that is
\begin{align}
\nonumber
    \hat V (\hat F \otimes \hat F)\sket{\psi^{\rm c}_{m_1},\psi^{\rm c}_{m_2}}  & =  \ri^{m_1+m_2} \hat V \sket{\psi^{\rm c}_{m_1},\psi^{\rm c}_{m_2}} \\
    & = \ri^{m_1+m_2} \sum_{k=0}^\infty \frac{(\ri \vartheta)^k (\position_1 \position_2)^k}{k!}\sket{\psi^{\rm c}_{m_1},\psi^{\rm c}_{m_2}} \\
    & = \sum_{m'_1,m'_2=0}^\infty \gamma_{m'_1,m'_2} \sket{\psi^{\rm c}_{m'_1},\psi^{\rm c}_{m'_2}} \;,
\end{align}
since the Hermite functions are exact eigenfunctions of the Fourier transform $\hat F$ in $\mathbb R$, of eigenvalue $\ri^m$. The last line follows from the same recurrence relation in Eq.~\eqref{eq:cont-positionrecurrence}. Also, recall that the Fourier transform maps $\position\mapsto\momentum\mapsto-\position$, so that $\hat W = (\hat F^{-1} \otimes \hat F^{-1}) \hat V (\hat F \otimes \hat F)$. This in turn implies $ \alpha_{m'_1,m'_2}= (-\ri)^{m'_1+m'_2}\gamma_{m'_1,m'_2}$ in the lemma.

Consider now the sum in Eq.~\eqref{eq:VbarHermite}.
To complete the action of $\overline{W}$, we would need to act with $F^{-1} \otimes F^{-1}$. However, since the sum ranges over all $m'_1$ and $m'_2$, we cannot directly use the error bound in Eq.~\eqref{eq:Fouriereigenvalue}. To fix this, 
we note that if we were able to cutoff that sum so that $m'_1,m'_2 \le c L$, then Eq.~\eqref{eq:Fouriereigenvalue} and Eq.~\eqref{eq:VbarHermite} would imply
\begin{align}
\nonumber
    \overline{W} \sket{\psi_{m_1},\psi_{m_2}} &=(F^{-1} \otimes F^{-1})
    \sum_{m'_1,m'_2=0}^{cL}\gamma_{m'_1,m'_2} \sket{\psi_{m'_1},\psi_{m'_2}} + \delta + \cO(\exp(-\gamma L))\\ 
    \nonumber
    &=\sum_{m'_1,m'_2=0}^{cL}
    (-\ri)^{m'_1+m'_2}\gamma_{m'_1,m'_2} \sket{\psi_{m'_1},\psi_{m'_2}} +\delta
    + \cO(L^2 )\times \cO(\exp(-\gamma L))
    \\
    \nonumber & =\sum_{m'_1,m'_2=0}^\infty
    (-\ri)^{m'_1+m'_2}\gamma_{m'_1,m'_2} \sket{\psi_{m'_1},\psi_{m'_2}} +2\delta +  \cO(\exp(-\gamma ' L)) \\
    & = \sum_{m'_1,m'_2=0}^\infty
    \alpha_{m'_1,m'_2} \sket{\psi_{m'_1},\psi_{m'_2}} +2\delta +  \cO(\exp(-\gamma ' L))\; ,\label{eq:w-bar-error}
\end{align}
where $\delta$ is the error due to cutting off the sum at $m'_1,m'_2 \le cL$.
The term $ \cO(L^2) \times \cO(\exp(-\gamma L))$ is 
a worst case bound coming from adding the $\cO(L^2)$
errors in the action of the Fourier transform. It is 
exponentially small in $L$, giving the constant $\gamma'>0$.

Next we show that the contribution $\delta$ in Eq.~\eqref{eq:w-bar-error} can also be made exponentially small. This error satisfies
\begin{align}
    \delta &\le \left\| \sum_{m'_1 \ge c L,m'_2 \ge c L} (-\ri)^{m'_1+m'_2}\gamma_{m'_1,m'_2} \sket{\psi_{m'_1},\psi_{m'_2}}\right\| \\
    &\le
    \label{eq:deltabound}\left| \sum_{m'_1 \ge c L,m'_2 \ge c L}  \gamma_{m'_1,m'_2}\right| \times \cO(\sqrt L) \;,
\end{align}
since $\|\dshostate{m}\|=\cO(L^{1/4})$ for all $m\ge cL$, which follows from $|\psi_m(x)|\le 1$ for the Hermite functions. The inequality now involves the same coefficients obtained by the action of $\hat V$. Note that
from the recurrence
\begin{align}
    \|(\position)^k \tshostate{m}\| \le \sqrt{m+1}
      \|(\position)^{k-1} \tshostate{m+1}\| \le \sqrt{(m+1)\ldots(m+k)} \;.
\end{align}
Then, if $m_1 \le cL/2$ and $m_2 \le cL/2$,
and $k \ge cL$,
\begin{align}
  \left\|  \frac{(\ri \vartheta)^k (\position_1 \position_2)^k}{k!} \ket{\psi_{m_1}^{\rm c},\psi_{m_2}^{\rm c}}
  \right \| &\le \frac{|\vartheta|^k}{k!}
  \sqrt{(m_1+1)\ldots(m_1+k)(m_2+1)\ldots(m_2+k)} \\
  & \le\frac{|\vartheta|^k}{k!} (cL/2+1) \ldots (c'L+k) \\
  &\le \frac{|\vartheta|^k(c L/2 + k)^k}{k!} \\
  &\le \left(\frac{e|\vartheta|(c L/2 + k)}{k}\right)^k & \text{(Stirling's)} \\
  &\leq (3/4)^k \;.
\end{align}
Here we have used the fact that $e |\vartheta| \le 1/2$.
The terms in Eq.~\eqref{eq:deltabound} are indeed generated by the action of $(\position_1 \position_2)^k$ for $k \ge cL$.
Thus we have
\begin{align}
     \delta & \le   \left| \sum_{m'_1 \ge cL,m'_2 \ge cL}  \gamma_{m'_1,m'_2}\right| \times \cO(\sqrt L) \\
     & \le \sum_{k=c L}^{\infty} (3/4)^k \times \cO(\sqrt L) \\
     & \le \exp(-\gamma'' L)
\end{align}
where the constant $\gamma''>0$ can be easily obtained from the inequalities. 

Combining all results we have shown the exponentially small error
 \begin{align}
       \left\| \overline W \sket{\psi_{m_1},\psi_{m_2}} - \sum_{m'_1,m'_2=0}^\infty
    \alpha_{m'_1m'_2}\sket{\psi_{m'_1},\psi_{m'_2}} \right\| \le  \exp(-\tilde \gamma L)
    \end{align}
    where the positive constant $\tilde \gamma>0$
    can also be obtained from the inequalities.
\end{proof}

As a note, the numerical results of Ref.~\cite{somma2016}
show that $c \ge 3/4$, and hence $c' \ge 3/8$ for this Lemma. Furthermore, 
a similar result on an exponentially small discretization error applies to exponentials involving $\position_j \momentum_k$, if they are replaced by exponentials of $\discreteposition_j \discretemomentum_k$.
This is necessary to include all the exponentials appearing in Lemma~\ref{lem:FFQHOs} or Eq.~\eqref{eq:Uinoscillatorsrep}.
We remark that the only such terms satisfy $j \neq k$, allowing us to commute the position and momentum operators and obtain a very similar bound to Lemma~\ref{lem:simple-exponential-bound-p}.

\begin{lemma}\label{lem:simple-exponential-bound-xp}
    Consider the simpler unitary $\hat W= \exp(\ri \vartheta \hat x_j \hat p_k)$ acting on two quantum harmonic oscillators, and the corresponding discretization $\overline W = \exp(\ri \vartheta \overline x_j \overline p_k)$. Let the phase be $\vartheta \in [-1/2e, 1/2e]$. Then, 
    if $\hat W \tshostate{m_j,m_k}= \sum_{m_j^\prime ,m_k^\prime =0}^\infty
    \alpha_{m_j^\prime m_k^\prime }\tshostate{m_j^\prime ,m_k^\prime }$, 
    we obtain
     \begin{align}
       \| \overline W \dshostate{m_j,m_k} - \sum_{m_j^\prime ,m_k^\prime =0}^\infty
    \alpha_{m_j^\prime m_k^\prime }\dshostate{m_j^\prime ,m_k^\prime } \| \le \exp(-\tilde \gamma L)
    \end{align}
    for all $m_j, m_k \le c'L$, where $c'<1$
    and $\tilde \gamma >0$ are positive constants.
\end{lemma}

The proof steps are essentially the same as those in Lemma~\ref{lem:simple-exponential-bound-p}.

We are now ready to prove the main result of this section, namely Lemma~\ref{lem:matrix-element-closeness}.

\begin{proof}
The unitary $\hat U$ of Eq.~\eqref{eq:Uinoscillatorsrep}
is a product of exponentials of monomials quadratic in $\position_j$ and $\momentum_j$. The number of terms $r$ is quadratic in $n$, i.e., a constant for constant $n$. 
Each of these exponentials is replaced by the discrete version. We can naively add up the error accrued by each term of this product by applying Lemmas~\ref{lem:simple-exponential-bound-x},~\ref{lem:simple-exponential-bound-p}, and~\ref{lem:simple-exponential-bound-xp}, but 
this does not suffice.

Recall that these lemmas apply to a Fock input state, and after the action of the first term in the product, the state can now be a linear combination of infinitely many Fock states. 
The hypothesis of Lemmas~\ref{lem:simple-exponential-bound-x},~\ref{lem:simple-exponential-bound-p}, and~\ref{lem:simple-exponential-bound-xp}, however, require the $m_i$'s to be sufficiently small, and each of these exponentials can create Fock states where the $m'_i$'s can be arbitrarily large. 
Nonetheless, in the proof of Lemma~\ref{lem:simple-exponential-bound-p}, we have shown that it is possible 
to perform Taylor series expansions of the exponentials and drop terms with sufficiently large $k$. In particular, we saw that it suffices if $k \ge 2c'L$, if the input state has $m_i \le c'L$. 
In the context of the unitary $\hat U$, 
which involves a product of $r=\cO(n^2)$ exponentials,
this means that if the initial state is $\ket{\psi_{m_1}^{\rm c},\ldots,\psi_{m_n}^{\rm c}}$ with $m_i \le M \le c' L/2^r < c'L$, 
then it suffices to cutoff the Taylor series of the first exponential at $c'L/2^{r-1}$. The state after the cutoff, which is the input state to the second exponential, is a linear combination of Fock states with $m_i \le c' L/2^{r-1}$.
We can then cutoff
the second exponential at $c' L/2^{r-2} \le c' L$, and so on until the last exponential at $c'L$.

Applying a union bound over all $r =\cO(n^2)$ terms in $\hat U$ and using the trivial bound where $m_i \leq c' L$ for all steps, the overall error 
incurred in all these truncations 
would be at most
\begin{align}
    \cO (n^2 (c^\prime L)^{n/2} \times \exp(-\tilde \gamma L)) \;.
\end{align}
This is the error of approximating each $\exp(\ri \vartheta_b \hat O_b)$ in $\hat U$ by the corresponding truncated Taylor series.
We can then invoke    Lemmas~\ref{lem:simple-exponential-bound-x},~\ref{lem:simple-exponential-bound-p}, and~\ref{lem:simple-exponential-bound-xp} to replace these Taylor series by those of the finite-dimensional operators, and then approximate those series by the corresponding finite-dimensional unitaries  $\exp(\ri \vartheta_b \overline O_b)$ that define $\overline U$.
This is possible
since the unitaries are now acting on Fock states that have bounded $m_i \le c' L$.
This will bring an additional error term 
$\cO (n^2 (c^\prime L)^{n/2}\exp(-\tilde \gamma L))$.
Assuming $n$ to be a constant, we combine all these errors in the term $\exp(-\overline \gamma L)$, for a properly defined $\overline \gamma >0$.

Moreover, in each of the Lemmas we assume an upper bound on $m_i$. In particular, for the 
first exponential we have $M \le c' L/2^r$, placing the lower bound $L \ge 2^r M/c'$.

\end{proof}

\section{Isometry}
\label{sec:isometry}
We now discuss steps 1 and 4 of  Sec.~\ref{sec:mainresults}. That is, we give a procedure to perform the descending lexicographic ordering step and outline our application of the Hermite transform.
Together, these steps will allow us to implement the isometry mapping the standard basis into the more convenient Fock basis of Hermite states.

\subsection{Descending lexicographic ordering}
\label{sec:lexicographic}

For given $n \ge 2$ and $M \ge 0$, the first step in Sec.~\ref{sec:mainresults}
concerns a unitary that  maps basis states as
\begin{align}
  V_1  \ket \ell \mapsto \ket{m_1(\ell),\ldots,m_n(\ell)} \; ;
\end{align}
in particular $\ket 0 \mapsto \ket {M,0,0,\ldots,0}$,
$\ket 1 \mapsto \ket{M-1,1,0,\ldots}$ and so on
until $\ket {N-1}\mapsto \ket{0,0,0,\ldots,M}$.
Recall that the dimension satisfies $N=\binom{M+n-1}{n-1}$.
This $V_1$ performs then a descending lexicographic ordering,
with the constraint that $\sum_{i=1}^n m_i = M$
and $m_i \ge 0$.

The functions $m_i(\ell)$ can computed efficiently and the same applies to the inverse function that satisfies
$\bm^{-1}(m_1(\ell),\ldots,m_n(\ell))=\ell$. 
Then we can use reversible circuits for these computations
and perform the desired mapping  coherently. 
The standard strategy is to first compute $\ket{\ell} \mapsto \ket{\ell} \ket{m_1(\ell),\ldots,m_n(\ell)}$ and then uncompute the 
$\ket{\ell}$ register by reversibly computing $\bm^{-1}$
and performing the map $\ket{\ell} \ket{m_1(\ell),\ldots,m_n(\ell)} \mapsto \ket{m_1(\ell),\ldots,m_n(\ell)}$. The overall complexity is that of the computation plus the uncomputation steps.

\begin{lemma}
    There is a quantum circuit which implements the map $V_1$ with gate complexity $\cO(n^2 \log^3 M)$.
\end{lemma}

The proof is contained in Appendix~\ref{app:lexicographicranking}. 
There we first show how given
$\ell \in \{0,\ldots,N-1\}$,
we can obtain the digits $m_1(\ell),\ldots,m_n(\ell)$ such that ${\rm rank}_{\rm desc}(m_1(\ell),\ldots,m_n(\ell))=\ell$,
where ${\rm rank}_{\rm desc}$ is the ranking. 
This reduces to a partitioning task where we can recursively compute the digits $m_1(\ell)$ all the way to $m_{n-1}(\ell)$. The last digit is determined by the constraint $\sum_{i=1}^n m_i = M$. Each digit is obtained via a binary search approach, based on counting the number of partitions before any string that starts with $m_1,\ldots,m_j$.  Next we show how given the digits $m_1,\ldots,m_n$, we can compute ${\rm rank}_{\rm desc}(m_1,\ldots,m_n)$. To this end we efficiently compute how many
configurations have the first digit less than $m_1$, then how many configurations have the first digit $m_1$ and the next less than $m_2$, and so on. Adding all these configurations determines the ranking. The complexity of the approach is mostly dominated by that of computing corresponding binomial coefficients, and can be shown to be $\cO(n^2 \log^3 (M))$ in the worst case.

\subsection{Quantum Hermite transform}
\label{sec:QHT}

To complete the isometry to  map basis states into Hermite states of $n$ quantum harmonic oscillators, we will use the recent quantum Hermite transform (QHT) of Ref.~\cite{jain2025efficient}.

\begin{theorem}[Efficient quantum Hermite transform, adapted from Ref.~\cite{jain2025efficient}]
\label{thm:QHT}
    Let $M>0$ be the dimension of the subspace for the Hermite transform
and $\epsilon >0$ the error. Then, there
 exists a quantum circuit of complexity $\cO((\log M +
\log(1/\epsilon))^3 \times \log(1/\epsilon))$ that can perform the following map with error $\epsilon$:
\begin{align}
\label{eq:QHT}
     \sum_{m=0}^{M-1} \alpha_m \ket m \mapsto 
      \sum_{m=0}^{M-1} \alpha_m \ket {\psi_m}\;.
\end{align}
The coefficients $\alpha_m \in \mathbb C$ are arbitrary and normalized, i.e., $\sum_m |\alpha_m|^2=1$. The Hermite states  
$\ket {\psi_m}$ are the ones described in Sec.~\ref{sec:finiteQHO} and of dimension $L=\cO(M^{2.25}/\epsilon^{3.25})$.
\end{theorem}

In Ref.~\cite{jain2025efficient}, the QHT is simply a basis change and hence the Hermite states
can be defined up to a phase. Nevertheless, it is simple to show that the resulting phase of that construction is exactly +1 for all $m$ when the Hermite states $\ket {\psi_m}$ in Eq.~\eqref{eq:QHT} 
are according to the definition of Sec.~\ref{sec:finiteQHO}. If this were not the case, then the unitary $\hat U$ of $SU(n)$ in the oscillator representation, corresponding to $U$, would be the desired one up to conjugation with a diagonal unitary that contains the phases. 
That is, we would not be implementing $U$ 
exactly with our quantum algorithm as specified in Secs.~\ref{sec:problemstatement} and~\ref{sec:oscillatorrep}, but rather $U$ up to this conjugation. The resulting transformation would still be a valid irrep of $SU(n)$ but not exactly the one we want.

The construction of the QHT in Ref.~\cite{jain2025efficient} follows simple state preparation steps. For each $m$, first it maps $\ket m \mapsto \ket m \ket {\phi_m}$, where 
the $\ket {\phi_m}$ have constant overlap with the $\ket{\psi_m}$. This can be done using 
some simple oscillatory approximations to the Hermite functions. Next it maps $\ket{\phi_m}\mapsto \ket{\psi_m}$ using a combination of quantum phase estimation (eigenstate filtering) and fixed point amplitude amplification. Then it maps $\ket m \ket{\psi_m}\mapsto \ket 0 \ket{\psi_m}$ to uncompute the $\ket m$ register, also using quantum phase estimation. These steps perform the right operation on arbitrary states due to linearity. The efficient complexity results from the ability to prepare the $\ket{\phi_m}$ efficiently and a corresponding fast forwarding result that enables efficient quantum phase estimation. The dimension requirement for $L$ being polynomial in $M$ and $1/\epsilon$ is mainly
due to a proof that shows constant overlap between $\ket{\phi_m}$ and $\dshostate{m}$, but might be improved with better analyses.

Then, we can complete the isometry for the $SU(n)$ transform by $n$ uses of the QHT, each within error $\epsilon/n$ and corresponding dimension $L$. 
That is,
\begin{align}
    ({\rm QHT})^{\otimes n} \ket{m_1,\ldots,m_n} \mapsto
    \ket{\psi_{m_1},\ldots,\psi_{m_n}}
\end{align}
within error $\epsilon$.
This is the unitary $V_2$ in step 2 of Sec.~\ref{sec:mainresults}.
Note that we are interested in performing this map for up to $M$ Hermite states, and since $M<N$  (i.e., $N=\binom{M+n-1}{n-1}$), the resulting complexity is polylogarithmic in $N$ and $1/\epsilon$. Furthermore, in comparing the requirements on $L$ from Lemma~\ref{lem:matrix-element-closeness} and Thm.~\ref{thm:QHT}, a suitable choice of $L=\cO(N^{2.25}/\epsilon^{3.25})$ will suffice for overall error $\epsilon$.

\section{Quantum algorithm}
\label{sec:quantumSU(n)}

Our main quantum algorithm for the high-dimensional $SU(n)$ problem essentially follows the four steps in 
Sec.~\ref{sec:mainresults}.
Let $\epsilon>0$ be the error, $N>0$ be the dimension of the irrep determined via $N=\binom{M+n-1}{n-1}$, and $n \ge 2$ be the dimension of the group $SU(n)$.
     Let $\vartheta \in [0,4 \pi)$ be the phase given within $\cO(\log(N/\epsilon))$ bits of precision, $T \in \{H_i,S_{j,k},A_{j,k}\} \in \mathbb C^{N \times N}$
     be a Hermitian matrix in the $N$-dimensional representation of $\mathfrak{su}(n)$ as in Sec.~\ref{sec:problemstatement}, and consider the implementation of the exponential $\exp(\ri \vartheta T)$.
     We let $\hat T$ be the corresponding $T$ in the space of the $n$ quantum harmonic oscillators (replacing $H_i \mapsto \hat H_i$, $S_{j,k} \mapsto \hat S_{j,k}$, or  $A_{j,k} \mapsto \hat A_{j,k}$). Similarly, $\bar O_b$ is the corresponding $\hat O_b$ but in the space of the discrete quantum harmonic oscillator.

\begin{algorithm}
    \caption{Efficient high-dimensional $SU(n)$ circuit}  \label{alg:high-dim-sun-rotation} 
     Compute the local dimension $L=\cO(M^{2.25}/\epsilon^{3.25})$ and consider the space $(\mathbb C^L)^{\otimes n}$. \\
     Implement the lexicographic ordering unitary $V_1$
     as described in Sec.~\ref{sec:lexicographic}. \\
     Apply $V_2=\text{QHT}^{\otimes n}$ as described in Sec.~\ref{sec:QHT}. \\
     Factorize $\exp(\ri \vartheta \hat T)=\prod_{b=1}^{r} \exp(\ri \vartheta_b \hat O_b)$
     according to Lemma~\ref{lem:FFQHOs}.\\
    \For{$b \in [r]$}{
    Compute $t = \lceil 2e \vartheta_b \rceil$ and define $\vartheta_b^\prime = \vartheta_b/t$. \\
    Perform $t$ repetitions of $\exp(\ri \vartheta_b^\prime \overline O_b)$ to satisfy the conditions of Lemmas~\ref{lem:simple-exponential-bound-x} to~\ref{lem:simple-exponential-bound-xp}.
    }
      Implement the inverse of the $n$ QHTs and the inverse of $V_1$.
\end{algorithm}

This algorithm implements the desired $U$ as $(V_1)^\dagger (V_2)^\dagger\exp(\ri \vartheta \overline T)V_2 V_1$,  where $\exp(\ri \vartheta \overline T)=\prod_{b=1}^{r} \exp(\ri \vartheta_b \overline O_b)$.

\subsection{Complexity}
\label{sec:complexity}

We now prove our main result in Thm.~\ref{thm:main}.
Thus far we showed that the complexity of $V_1$ and the $n$ QHTs is polylogarithmic in $1/\epsilon$ and $N$.
More specifically, the isometry can be implemented with complexity
\begin{align}
    \cO(n^2 \times \log^3(N) + n((\log (M)+\log(1/\epsilon))^3 \times \log(1/\epsilon))\;.
\end{align}
For constant $n$, this is $\cO((\log (M)+\log(1/\epsilon))^3 \times \log(1/\epsilon))$ and dominated by the QHTs.
More generally, we can write $\cO(n^2((\log (N)+\log(1/\epsilon))^3 \times \log(1/\epsilon))$ for this complexity.

The additional complexity comes from the implementation
of $\exp(\ri \vartheta \overline T)$. For each diagonal unitary, such as an exponential of $\overline x_j \overline x_k$ and $\vartheta$ is given with $\cO(\log(N/\epsilon))$ bits of precision, the unitary   can be implemented with cost $\cO(\log^2(N/\epsilon))$ using standard techniques. If the operation is not diagonal, we need to add the cost of the centered Fourier transform of dimension $L$, which is $\cO(\log^2 L)$ or $\cO(\log^2(N /\epsilon))$. 

For a general $U$ in the high-dimensional irrep of $SU(n)$ as in Thm.~\ref{thm:main}, there will be $n^2-1$ phases and exponentials. Hence, the complexity of implementing 
the corresponding unitary in the oscillator representation will be $\cO(n^2 \log^2(N /\epsilon))$, since there are $\cO(n^2)$ terms in the product.

The resulting complexity for a general $U$ is then 
\begin{align}
     \cO(n^2(\log (N)+\log(1/\epsilon))^3 \times \log(1/\epsilon) + n^2 \log^2(N /\epsilon))
\end{align}
or simply $\cO(n^2(\log (N)+\log(1/\epsilon))^3 \times \log(1/\epsilon))$.
 This is the actual complexity in the result in Thm.~\ref{thm:main}.

\section{Applications}
\label{sec:applications}

We begin with a concrete application of the $SU(2)$ transform to optimal quantum expanders and later comment on potential applications of the transforms to quantum simulation problems.  

\subsection{Explicit Ramanujan quantum expanders}

We provide the proof of Corollary~\ref{cor:expanders}.
A quantum expander is a generalization of a classical expander using quantum channels instead of Markov chains. A quantum expander's goal is to randomize an input quantum state quickly, acting like a depolarizing channel. They are useful in physics
and quantum information; for example, they can be 
be used to probe quantum chaos since they are fast scramblers, they play a role in matrix product states, and are related to constructions of unitary designs (Cf.~\cite{ben2007quantum,hastings2007random,harrow2007quantum}).
The definition of a quantum expander relies on the spectral properties of the channel. Let $\rho \in \mathbb C^{N \times N}$ be a density matrix and $\varepsilon(\cdot)$ be the quantum channel. 
Formally, we say $\varepsilon$ is an $(N,D,\lambda)$ quantum expander if it satisfies:
\begin{enumerate}[label=(\roman*)]
    \item $\varepsilon(\one_N)=\one_N$ is the unique fixed point,
    \item $\varepsilon(\rho)=\sum_{\rd=1}^D A_\rd^{} \rho A_\rd^\dagger$, where $A_\rd$ are Kraus operators, and
    \item the second largest eigenvalue of the $\varepsilon$ is at most $\lambda <1$. 
\end{enumerate}
This last condition is equivalent to 
\begin{align}
    \sup_{X \in \mathbb C^{N \times N}: \tr(X)=0} \|\varepsilon(X)\|_{\rm Fr}/\|X\|_{\rm Fr}\le \lambda < 1,
\end{align}
where $\|.\|_{\rm Fr}$ is the Frobenius norm. (We assume all other eigenvalues to be at most $\lambda$ in magnitude.) The degree of the expander is $D$ and   we are interested in cases where $D$ is also a constant independent of $N$. Intuitively, the larger the gap $1-\lambda$ between the two largest eigenvalues is, the faster the expander scrambles. Optimal quantum expanders, also known as Ramanujan quantum expanders, are then those that satisfy an optimal relationship between $\lambda$ and $D$, namely $\lambda \le 2 \sqrt{n-1}/D$. If this eigenvalue relationship is not satisfied exactly, but only satisfied up to a very small and arbitrary difference $\epsilon > 0$, we refer to such quantum expanders as approximate Ramanujan quantum expanders.

It turns out that exact Ramanujan quantum expanders can be explicitly constructed from certain $SU(2)$ unitaries and, since we can implement these unitaries with arbitrary accuracy, our construction gives approximate {\em and} efficient Ramanujan quantum expanders for $D \ge 2$. That is, of interest to quantum computing are those quantum expanders $\varepsilon$ that can be simulated with complexity that is $\polylog(N)$. 
In particular, Refs.~\cite{gamburd1999spectra,bourgain2006new}
provide a result on quantum expanders from $SU(2)$: there is a particular finite set $\{U_1,\ldots,U_D\}$ of unitaries in the $N$-dimensional representation of $SU(2)$ such that the quantum channel $\varepsilon(\rho)=\frac 1 D \sum_\rd U_\rd^{} \rho U_\rd^\dagger$ is a quantum expander. Furthermore, it is possible to construct {\em exact} Ramanujan quantum expanders from $SU(2)$ unitaries derived from the LPS construction~\cite{lubotzky1988ramanujan,harrow2007quantum}. However, prior to this work, an efficient implementation of these unitaries was not known.

Our efficient $SU(n)$ transform can then be used to construct such efficient and (approximate) Ramanujan quantum expanders. This construction is explicit and specified by a sequence of rotation angles and axes that determine the $U_\rd$'s; see Appendix~\ref{app:Ramanujanexpander} for details. This contrasts other constructions of approximate Ramanujan quantum expanders based on random quantum circuits. For example, Ref.~\cite{hastings2007random} shows that approximate Ramanujan quantum expanders can be obtained from $D \ge 2$ Haar random unitaries. In that case it is possible to show that the eigenvalue relationship is satisfied up to an additive term $\cO(1/\sqrt N)$, with overwhelming probability $1-\cO(\exp(-cN)$ for some $c>0$. However, this construction is not efficient since Haar random circuits need exponentially many gates\footnote{In addition we note that, while the additive term to the eigenvalue relationship is small, it depends on $N$ and cannot be made arbitrarily small for fixed $N$. In our $SU(2)$ constructions, we can take $\epsilon \rightarrow 0$ independently of $N$.}. To bypass this efficiency problem, Refs.~\cite{brandao2016local,hunter2019unitary,haferkamp2021improved} 
investigate the related problem of constructing $t$-designs from efficient random quantum circuits that are local or have all-to-all connectivity. Those results suggest that quantum expanders can be constructed from random quantum circuits if the depth is sufficiently large, but they are not known to  be Ramanujan, even approximately. Furthermore, for $D$ constant, those constructions require using $\Omega( \polylog N)$ random bits  for the random gates, while constructions based on $SU(2)$ only require a constant number of random bits to implement the channel. 

Another approach for constructing Ramanujan quantum expanders might be to use recent efficient constructions of $t$-designs in extremely low depth~\cite{SchusterHaferkampHuang25} to try to reproduce the pseudorandomness needed in the proof that Haar random unitaries are Ramanujan quantum expanders~\cite{hastings2007random}. However, to the best of our knowledge, the eigenvalue relation property uses $\poly(N)$ moments of the Haar random distribution (as written they use $\Theta(N^{2/15})$ moments). It's possible that the analysis could be refined to get $\varepsilon$-close to the Ramanujan bound with fewer moments, but this is not straightforward. It is known that $t$-designs require depth $t$ circuits, hence the circuit complexity of such an approach would necessarily be $\poly(N)$ without further insight.

In Appendix~\ref{app:Ramanujanexpander} we give a concrete example for a Ramanujan quantum expander based in $SU(2)$, for arbitrary dimension $N \ge 2$ and certain degrees $D \ge 4$. We also work out the specific case $D=6$, where the second largest eigenvalue satisfies $\lambda \le 2 \sqrt 5/6$ for all $N \ge 2$. 
Setting the rotation angle to $\theta=2 \arccos(1/\sqrt 5)$, we define the quantum channel $\varepsilon_{\mathfrak{su}(2)}$ that samples a random unitary $U^\pm_a = \exp(\pm \ri \theta J_a)$, where $a=x,y,z$ and $J_a$
is the angular momentum operator. In Fig.~\ref{fig:spectralgapexpander}
we show the spectral gap of this quantum channel and how it satisfies   the optimal bound. In combination with our results that prove these unitaries can be implemented with complexity $\polylog(N)$ and $\polylog(1/\epsilon)$, the result is an efficient and explicit approximate Ramanujan quantum expander.
\begin{figure}[htb]
    \centering
    \includegraphics[scale=0.4]{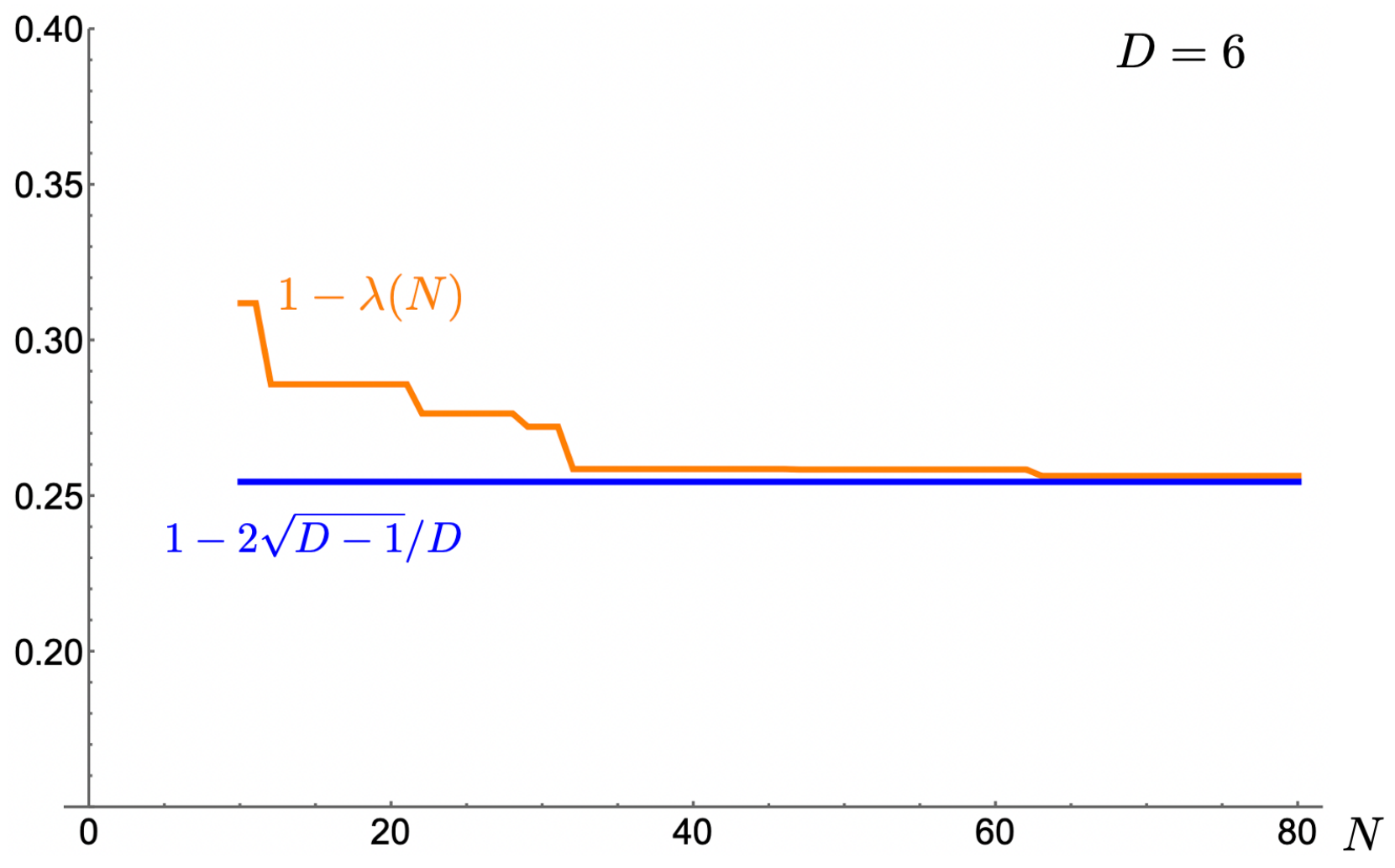}
    \caption{The spectral gap of the quantum expander $\varepsilon_{\mathfrak{su}(2)}$ as a function of the dimension $N \ge 10$ of the representation. The spectral gap satisfies the optimal bound $1-\lambda(N) \ge 1-2\sqrt 5/6$.}
    \label{fig:spectralgapexpander}
\end{figure}

  \subsection{Fast-forwarding quantum simulation}

We can also consider applications to quantum simulation. Within the context of simulating quantum evolution, our quantum circuits for $SU(n)$ transforms give yet other examples to the growing literature of  fast-forwarding~(Cf. \cite{atia2017fast,babbush2018low,gu2021fast,jain2025efficient,Apel2026}). In general, for simulating a Hamiltonian $H$ for time $t$, the complexity is $\Omega(\|H\|t)$, where $\|H\|$ is some norm like the spectral norm. In our case, where the Hamiltonian is spanned by the corresponding representation of the Lie algebra 
$\mathfrak{su}(n)$, the spectral norm can satisfy $\|H\|=\Theta(N)$, yet our quantum transforms allow us to simulate this evolution exponentially faster, in time $\polylog(N)$. 
A main example is the one underlying our construction, where we readily showed how to fast forward the dynamics of $n$ quantum harmonic oscillators, where the Hamiltonian takes the form
\begin{align}
    H = \sum_{j,k=1}^n \alpha_{j,k} a^\dagger_j a^{}_k \;.
\end{align}
While this Hamiltonian acts on the infinite-dimensional spaces, we also showed that with proper discretizations and dimension cutoff, we can recover the properties in the continuum with excellent accuracy. This result, however, does not necessitate the quantum Hermite transform since it is a bosonic system by definition and the evolution is a transformation over the Hermite states. We showed how to do this efficiently via the Jordan-Wigner representation.

Another interesting example is that of the quantum kicked top, which is a model used to study quantum chaos, and was a motivation for Ref.~\cite{zalka2004implementing}. 
In this case, the periodic and time-dependent Hamiltonian is $H(t) = \gamma J_y + \beta (J_z)^2 \sum_n \delta (t- n T)$, where $T$ is the period, and $J_a$ the angular momentum operators in the $N$-dimensional representation of $\mathfrak{su}(2)$. The evolution $U(t)$ induced by $H(t)$ can be expressed in terms of the Floquet operator $V=e^{-\ri \beta (J_z)^2} e^{-\ri \gamma J_y}$; that is, $U(kT)=V^k$.
Then, this can also be fast forwarded since: 
i) our results show how $e^{-\ri \gamma J_y}$ can be fast forwarded, and ii) the operator $e^{-\ri \beta (J_z)^2}$ is a diagonal unitary, which can implemented efficiently using standard techniques for diagonal unitaries.

\section{Conclusion \& open questions}
\label{sec:conclusions}

We provided efficient quantum circuits to implement unitaries in the $N$-dimensional representations of the group $SU(n)$. The circuits leverage a recently proposed quantum Hermite transform (QHT), and the Jordan-Schwinger representation of the Lie algebra $\mathfrak{su}(n)$ using $n$ quantum harmonic oscillators. To obtain complexity polylogarithmic in the dimension $N$, we showed that the dynamics of the $n$ quantum harmonic oscillators induced by the corresponding (quadratic) operators can be fast forwarded, and then we proved a similar result upon discretizing these systems. Our findings work for the totally symmetric representations, and we leave open the case of implementing $SU(n)$ unitaries using other representations.

We presented applications to efficient and explicit Ramanujan quantum expanders and to fast-forwarding time dynamics of specific quantum systems. 
In those cases, we assumed $n \ge 2$ to be constant, but also notice that our results can be extended to the case where $n=\polylog(N)$ in a straightforward way. 
This opens up the possibility of fast-forwarding other quantum systems.
Beyond, we expect that our results will unlock other novel applications.

We note that our quantum circuits 
might not be optimal in terms of gate count, which is an important problem when 
investigating actual practical applications. 
The complexity of the high-dimensional $SU(n)$
unitary is dominated by that of the QHT in our case, which in turn is dominated by some state preparation steps requiring reversible and coherent arithmetics. Hence, any significant improvement 
might result from improvements to the corresponding steps of the QHT. Finally, we end with some open questions for future work.

\begin{itemize}

    \item Can we use our construction to implement the Fourier transform for the groups $SU(n)$ efficiently? 

    \item Given that we only implement the totally symmetric irreps, can we also implement all the other high-dimensional irreps of $SU(n)$ efficiently?
    
    \item Can we construct explicit Ramanujan quantum expanders for all degrees and dimensions? This would be a quantum analog of a celebrated line of work~\cite{MSSone,MSSfour,Alon2021}.

    \item What are the optimized gate counts for implementing $N$-dimensional irreps of $SU(2)$ for reasonably large values of $N$?
\end{itemize}

\subsection*{Acknowledgements}

We thank Aram Harrow for 
emphasizing the application to explicit Ramanujan quantum expanders and Dave Bacon for discussions.
VI is supported by a National Science Foundation Graduate Research Fellowship. S. Jain is supported by an Amazon AI Fellowship.

\appendix

\section{Irreducible representations of the Lie algebra $\mathfrak {su}(n)$ using bosonic operators}
\label{app:schwingerrep}

The oscillator representation of 
$\mathfrak {su}(n)$, also known as the Jordan-Schwinger representation or mapping, provides a path 
to generating the $\mathfrak {su}(n)$ algebra
from the bosonic operators associated with distinct bosonic modes or distinct quantum harmonic oscillators.

Consider the creation and annihilation 
operators $a^\dagger_j$ and  $a^{}_j$, respectively, satisfying the commutation relations
\begin{align}
\label{eq:bosonicrelations}
    [ a^{}_j, a^\dagger_k] & =\delta_{j,k},  \;
     [a^\dagger_j, a^\dagger_k]  =0 , \;
      [a^{}_j, a^{}_k]  =0 ,
\end{align}
 where $1 \le j,k \le n$. We define the operators
 $\hat E_{j,k}:=a^\dagger_j a^{}_k$, being quadratic in bosonic operators and also ``number preserving''. That is, while $a^\dagger_j$ creates a boson in the $j^{\rm th}$ bosonic mode, $a^{}_k$ annihilates a boson in the 
 $k^{\rm th}$ bosonic mode. Accordingly, we define
 \begin{align}
     \hat S_{j,k}&:=\frac 1 2( \hat E_{j,k} + \hat E_{k,j} )\; , \\
     \hat A_{j,k}&:= \frac {\ri} 2(\hat E_{j,k} - \hat E_{k,j}) \; ,
 \end{align}
 where $1 \le j < k \le n$, and $\hat E_{k,j}=(\hat E_{j,k})^\dagger=a^\dagger_k a^{}_j$.
 These are Hermitian and, for $1 \le j < k \le n$, there are $n(n-1)$ such operators.
We also define the diagonal operators
\begin{align}
    \hat H_i := \frac 1 2 (a^\dagger_i a^{}_i - a^\dagger_{i+1} a^{}_{i+1})\;, 
\end{align}
where $1 \le i < n$. These are also Hermitian,  number preserving, and there are $n-1$ such operators.

The connection with $\mathfrak{su}(n)$ is readily apparent. Indeed, it is simple to show that $\{ \overline H_j,  \overline S_{j,k},  \overline A_{j,k} \}$ also span the real Lie algebra $\mathfrak{su}(n)$ from the commutation relations in Eq.~\eqref{eq:bosonicrelations}, which imply
\begin{align}
    [\hat E_{j,k},\hat E_{l,m}]&=\delta_{k,l} \hat E_{j,m} - \delta_{j,m} \hat E_{l,k} \;,\\
     [\hat E_{j,k},\hat H_i] &=\delta_{k,i} \hat E_{j,i}-\delta_{j,i} \hat E_{i,k} - \delta_{k,i+1} \hat E_{j,i+1} + \delta_{j,i+1} \hat E_{i+1,k}\;,\\
      [\hat H_i,\hat H_{i'}] & = 0 
    \;.
\end{align}
The structure factors of these commutations match those
of the commutation relations between the $n \times n$ matrices $E_{j,k}$, which readily provided an irreducible representation of $\mathfrak{su}(n)$.
Hence, the operators $\hat S_{j,k}$, $\hat A_{j,k}$, and $\hat H_i$
specify an $\mathfrak {su}(n)$ Lie algebra as well.

The next step is to construct irreducible representations of
$\mathfrak {su}(n)$ from these operators. To this end, consider $n$ different bosonic modes, where the total number of bosons is $M \ge 0$ and fixed. This subspace is spanned by the following Fock states:
\begin{align}
 \cB_{n,M} :=\left \{  \ket{m_1,m_2,\ldots,m_n} : \sum_{j=1}^n m_j = M \right \} \;.
\end{align}
A useful ordering for Fock states is the descending lexicographic ordering discussed in Sec.~\ref{sec:oscillatorrep}.
   Note that the action of the bosonic operators is such that
   \begin{align}
       \hat E_{j,k} \ket{m_1,\ldots,m_j, \ldots, m_k, \ldots ,m_n} = \sqrt{(m_j+1)m_k}\ket{m_1,\ldots,m_j+1, \ldots, m_k-1, \ldots, m_n}\;.
   \end{align} 
   This action automatically gives a representation of $\mathfrak{su}(n)$ using Fock states. Furthermore, the action is irreducible since there is always a sequence of operators $\hat E_{j,k}$ that connects any two arbitrary Fock states of $M$ bosons. The overall dimension of this irreducible representation $N$ is given by the composition of $M$ with $n$ non-negative integers. The formula is
   \begin{align}
       N(n,M) = \begin{pmatrix}
           M + n -1 \cr n-1
       \end{pmatrix}\;,
   \end{align}
   which can be obtained from the recursion $N(n,M)=\sum_{m=0}^M N(n-1,m)$ and $N(1,m)=1$. This dimension is not arbitrary unless $n=2$, in which case $N(n,M)=M+1$, and we can 
   set $M$ accordingly. More generally, these dimensions refer to the ``totally symmetric'' representations of $SU(n)$. For $n \ge 3$, these are a strict subset of all representations.

   We can illustrate this with an example. Consider the case $n=3$ and $M=2$. The dimension of the representation is $N(3,2)=6$ and the Fock states that span this space, in descending lexicographic ordering, are
   \begin{align}
\cB_{3,2}=\{\ket{2,0,0},\ket{1,1,0},\ket{1,0,1},\ket{0,2,0},\ket{0,1,1},\ket{0,0,2} \}\;.
   \end{align}
Using this ordering we can compute the representations of the $\hat E_{j,k}$; for example
\begin{align}
   \hat E_{1,2} = \begin{pmatrix}
   0 & \sqrt 2 & 0 & 0 & 0 & 0 \cr
       0 & 0 & 0 & \sqrt 2 & 0 &0 \cr
       0 & 0 & 0 & 0 & 1 & 0 \cr
       0 & 0 & 0 & 0 & 0 & 0 \cr
       0 & 0 & 0 & 0 & 0 & 0 \cr
       0 & 0 & 0 & 0 & 0 & 0
   \end{pmatrix}, \; \hat H_1=
   \begin{pmatrix}
   1 & 0 & 0 & 0 & 0 & 0 \cr
       0 & 0 & 0 & 0 & 0 &0 \cr
       0 & 0 & 1/2 & 0 & 0 & 0 \cr
       0 & 0 & 0 & -1 & 0 & 0 \cr
       0 & 0 & 0 & 0 & -1/2 & 0 \cr
       0 & 0 & 0 & 0 & 0 & 0
   \end{pmatrix}.
\end{align}

\subsection{High-dimensional representation of $\mathfrak{su}(2)$}
\label{app:su(2)lexicographic}

Another standard example is when $n=2$. In this case, we can obtain all the irreducible representations. For general $M$, considering the basis 
\begin{align}
    \cB_{2,M}=\{\ket{M,0},\ket{M-1,1},\ldots,\ket{0,M}\},
\end{align}
which is readily given in the descending lexicographic ordering,
we obtain
\begin{align}
   \hat E_{1,2} = \begin{pmatrix}
       0 & \sqrt M & 0 & \ldots & 0 \cr
       0  & 0 & \sqrt{(M-1)2} & \ldots & 0 \cr
       \vdots & \vdots & \vdots & \ddots & \vdots \cr
       0 & 0 & 0 & \ldots & \sqrt{M}\cr
       0 & 0 & 0 & \ldots & 0
   \end{pmatrix},
\end{align}
which is the well-known $\mathfrak{su}(2)$ raising angular momentum operator $J^+$, and also
\begin{align}
   \hat H_1 = \frac 1 2\begin{pmatrix}
       M & 0 & \cdots & 0 \cr
       0 & M-2 & \cdots & 0 \cr
       \vdots & \vdots & \ddots & \vdots \cr
       0 & 0 & \cdots & -M
   \end{pmatrix}  ,
\end{align}
which is   the angular momentum operator $J_z$.

\section{Ranking in the descending lexicographic ordering}
\label{app:lexicographicranking}

We would like to perform the mapping $\ket \ell \mapsto \ket{m_1(\ell),\ldots,m_n(\ell)}$ and its inverse,
by using coherent arithmetics. Next we describe 
the corresponding classical computation and uncomputation steps and their complexities, which also determine the quantum complexities.

\subsection{From an index to a lexicographically ordered partition}

Given $\ell \in \{0,\ldots,N-1\}$,
we want to obtain $m_1(\ell),\ldots,m_n(\ell)$, where $\sum_i m_i(\ell)=M$. 
This is a partitioning task, which can be thought of as placing $M$ balls into $k \le n$ bins, and there are $\binom{M + k - 1}{k - 1}$ ways of doing this. 
First we fix $M$ and $k=n-1$, and define $N_r$ to be the number of partitions with $r$ balls in the first bin. From the prior formula, we obtain
\begin{align}
N_r = \binom{M - r + k-1 }{ k-1}\;.
\end{align}
Now, define $S_r: = \sum_{k=r}^{M} N_k$. By a simple combinatorial argument, we have
\begin{align}
S_r = \sum_{r'=r}^{M} \binom{M - r' + k-1}{k-1} = \binom{M - r + k}{k}  \;.
\end{align}
This counts the number of partitions that are in descending lexicographic ordering before encountering \emph{any other} partition whose first element is $r-1$. 

Our algorithm to compute $m_1(\ell), \ldots,m_n(\ell)$ will proceed recursively: find the smallest $r$ such that $S_r \leq \ell$. This can be done via binary search in time $\log M$, since $S_r$ is monotonically decreasing in $r$. This will readily give $m_1(\ell)$.
Then place $r$ balls in the first bin and solve the task recursively with $\ell \to \ell - S_r$, $M \to M - r$, and $k \to k-1$, up to $k=1$. This will give
the first $n-1$ digits, and the digit $m_n(\ell)$
is automatically determined by the constraint $\sum_i m_i(\ell)=M$.

The worst-case complexity of this computation is $\cO(n \times \log(M) \times n \times \log^2 M)=\cO(n^2 \log^3 M)$. One factor of $n$ is from running $k=n-1$ to $k=1$.
Another factor $\log(M)$ is from binary search at each step. In each binary search we need to compute the binomial coefficient $\binom{M-r+k}{k}$, which involves $\cO(n)$ multiplications with bitstrings of size $\log^2(M)$, giving the final complexity.

\subsection{From a lexicographically ordered partition to an index}

To obtain the desired formula, it helps considering the ascending lexicographic ordering first. 
The ordering is such that 
the string $m_1,\ldots,m_n$ is bigger than $m'_1,\ldots,m'_n$
if there is a $j \in\{1,n\}$ such that $m_{j'}=m'_{j'}$
for all $j'<j$ and $m_j > m_j^\prime $. This allows us to rank the 
strings easily by considering each digit $m_i$ in ascending order. 
That is, a larger $m_1$ implies a larger ranking, determining the most relevant floor for the ranking. For given $m_1$, a larger $m_2$
also implies a larger ranking, determining the next relevant floor for the ranking, and so on. Our goal is to add the magnitudes of these floors and determine $\ell$.

More specifically, let $M \ge 0$ and $n \ge 2$ be fixed. 
The constrain is $m_i \ge 0$ and $\sum_i m_i = M$. 
Consider the string $m_1,\ldots,m_n$. If $m_1>0$
then we know in the ascending order that the strings with 
first digit $m'_1$ being $m_1-1,\ldots,0$ are ranked below this one. 
For a fixed $m_1'$, we need to count how to arrange $M'=M-m'_1$ into $d'=n-1$ bins. The formula for this is the one we already know, i.e., $\binom{M' +d'-1}{d'-1}$. This gives us the additive term in the ranking
\begin{align}
\nonumber
    r_1&=\sum_{m'_1=0}^{m_1-1}\binom{M-m'_1+n-2}{n-2} \\
    &= \binom{M+n-1}{n-1}-\binom{S_1+n-1}{n-1}\;,
\end{align}
where $S_1=M-m_1=\sum_{i>1} m_i$. In the context of bosonic states, this is the number of bosons to the right of $m_1$.

Next we need to rank according to $m_2$. That is, for fixed $m_1$ and $m_2$, we need to add the terms due to arranging $M'=M-m_1-m_2'$ into $d'=n-2$ bins. The formula is similar:
\begin{align}
\nonumber
    r_2&=\sum_{m'_2=0}^{m_2-1}\binom{M-m_1-m'_2+n-3}{n-3}
    \\
    &= \binom{M-m_1+n-2}{n-2}-\binom{S_2+n-2}{n-2}\;.
\end{align}
Here, $S_2 = \sum_{i>2}m_i$.
Next we keep adding the terms $r_3$ all the way up to $r_{n-1}$, since $r_n$ will be fixed from these due to the constraint $\sum_i m_i=M$.
The result is
\begin{align}
\nonumber
    {\rm rank}_{\rm asc}(m_1,\ldots,m_n) &=  \sum_{k=1}^{n-1} \binom{S_{k-1}+n-k}{n-k}-\binom{S_k+n-k}{n-k} \\
    \textbf{}
    & = N + \sum_{k=1}^{n-2}\binom{S_{k}+n-k-1}{n-k-1} -\sum_{k=1}^{n-1}\binom{S_k+n-k}{n-k} \\
   & = N-1 + \sum_{k=1}^{n-1}\binom{S_{k}+n-k-1}{n-k-1} - \binom{S_k+n-k}{n-k} \;,
\end{align}
where $S_0:=M$ and recall $N=\binom{M+n-1}{n-1}$. 
This can be computed classically in time $ \cO(n^2 \log^2 M)$ 
or $ \cO(n^2 \log^2 N)$ for  $n \ll N$, since there are $\cO(n)$ terms in the sum, each involving $\cO(n)$ multiplications with strings of size $\cO(\log N)$. Note that the exponent can be improved with faster than textbook multiplication. 

For the descending lexicographic ordering, 
we simply obtain
\begin{align}
\nonumber
    &{\rm rank}_{\rm desc}(m_1,\ldots,m_n) =(N-1)- {\rm rank}_{\rm asc} \\
      \label{eq:appinverselexicographic}
    & =    \sum_{k=1}^{n-1} \binom{S_k+n-k}{n-k}-\binom{S_{k}+n-k-1}{n-k-1}\;.
\end{align}
We can compute this also in time $ \cO(n^2 \log^2 M)$. Note that this is the desired function: $\ell(m_1,\ldots,m_n)\equiv {\rm rank}_{\rm desc}(m_1,\ldots,m_n)$.

\section{Ramanujan quantum expanders from $SU(2)$}
\label{app:Ramanujanexpander}

We discuss briefly the construction 
of optimal quantum expanders from $SU(2)$ transforms..
The LPS construction~\cite{lubotzky1988ramanujan,harrow02expander}
gives the rotation angles from the integer solutions to the quaternion equation
\begin{align}
\label{eq:appquaternion}
    a_0^2 + a_1^2 + a_2^2 + a_3^2 = p \;,
\end{align}
where $p$ is a prime and usually $p=1 \mod 4$.
That is, we are interested in integer quaternions of (squared) length $p$.
Here the degree is $D=p+1$, and the quantum channel is constructed from $D$ unitaries $U_1, \ldots, U_D$, so that $\varepsilon(\rho)= \frac 1 D \sum_{\rd=1}^D U_\rd^{} \rho U_\rd^\dagger$. Each of these unitaries can be explicitly written as
\begin{align}
    U_{\rd} = \exp (-\ri \theta^{\rd} (n_x^{\rd} J_x+ n_y^{\rd} J_y + n_z^{\rd} J_z))\;,
\end{align}
where $J_x$, $J_y$, and $J_z$ are the $SU(2)$ angular momentum operators, $0 \le \theta_\rd \le 2\pi$, and $\vec n^\rd =(n_x^\rd,n_y^\rd,n_z^\rd) \in \mathbb R^3$ satisfies $\|\vec n_\rd\|=1$. The rotation angles and axes satisfy
\begin{align}
    \theta^\rd = 2 \arccos \left(\frac {a_0^\rd}{\sqrt p} \right) , \; \vec n^\rd = \frac 1 {\sqrt{p-(a_0^\rd)^2}} (a_1^\rd,a_2^\rd,a_3^\rd),
\end{align}
where the supraindex $\rd$ refers to the $\rd^{\rm th}$ solution to Eq.~\eqref{eq:appquaternion}. (Note that many of the solutions are related via multiplication by imaginary units, so we only keep the $p+1$ ``distinct'' solutions.)

We illustrate this with an example for $p=5$, where $D=6$. 
The distinct solutions are
\begin{align}
    (1,2,0,0), (1,-2,0,0), (1,0,2,0), (1,0,-2,0), (1,0,0,2), (1,0,0,-2).
\end{align}
We compute all the rotation angles and axes:
\begin{align}
    \theta^1 & \simeq {126.87}^\circ  ,  \; \vec n^1 =  (1,0,0) \implies U_1 = \exp(-\ri  \theta^1 J_x )\;,\\
     \theta^2 & \simeq {126.87}^\circ  ,  \; \vec n^2 =  (-1,0,0) \implies U_2 = \exp(\ri  \theta^2 J_x )\; , \\
      \theta^3 & \simeq {126.87}^\circ  ,  \; \vec n^3 =  (0,1,0) \implies U_3 = \exp(-\ri  \theta^3 J_y )\\
       \theta^4 & \simeq {126.87}^\circ ,  \; \vec n^4 =  (0,-1,0) \implies U_4 = \exp(\ri  \theta^4 J_y )\;, \\
        \theta^5 &  \simeq {126.87}^\circ ,  \; \vec n^5 =  (0,0,1) \implies U_5 = \exp(-\ri  \theta^5 J_z )\; , \\
         \theta^6 &  \simeq {126.87}^\circ ,  \; \vec n^6 =  (0,0,-1) \implies U_6 = \exp(\ri  \theta^6 J_z )\;  .
\end{align}
All the angles are the same since $a_0=1$ in all these solutions. The result for the spectral gap of the superoperator is shown in Fig.~\ref{fig:spectralgapexpander} for this case.

Even for the case $p=3$, $D=4$, we can achieve optimal Ramanujan quantum expanders.
The relevant solutions to Eq.~\eqref{eq:appquaternion}
imply $a_0=0$ or $\theta^\rd=\pi$ for all $1 \le n \le 4$. The relevant rotation axes are $\vec n^1 = (1,1,1)/\sqrt 3$, $\vec n^2 = (1,1,-1)/\sqrt 3$, $\vec n^3 = (1,-1,1)/\sqrt 3$, and $\vec n^4 = (-1,1,1)/\sqrt 3$. Simulation results are in Fig.~\ref{fig:appexpander}.

\begin{figure}[htb]
    \centering
    \includegraphics[scale=0.35]{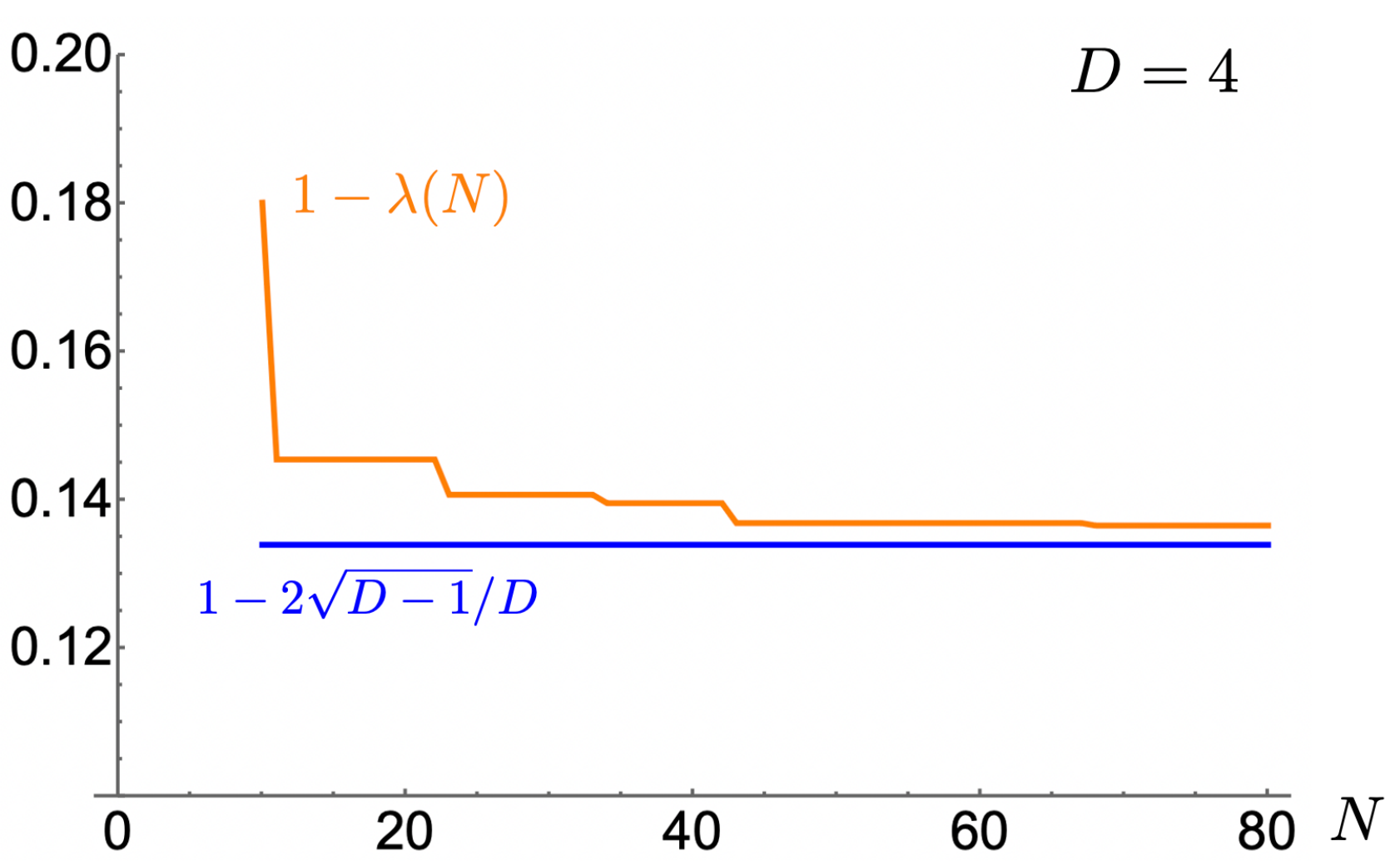}
    \caption{The spectral gap of the quantum expander obtained from Eq.~\eqref{eq:appquaternion} as a function of the dimension $N \ge 10$ of the representation for the $D=4$ case. The spectral gap satisfies the optimal bound $1-\lambda(N) \ge 1-2\sqrt 3/4$.}
    \label{fig:appexpander}
\end{figure}

\newpage
\DeclareUrlCommand{\Doi}{\urlstyle{sf}}
\renewcommand{\path}[1]{\small\Doi{#1}}
\renewcommand{\url}[1]{\href{#1}{\small\Doi{#1}}}
\bibliographystyle{plain}
\bibliography{refs}

\end{document}